\documentclass[aps,pra,showpacs,twoside,twocolumn,10pt]{revtex4-2}

\usepackage{graphicx,bm,amsmath,xcolor,braket,plain,color,amsthm,amsfonts,tabularx,graphicx,bbm,mathtools,esvect,wrapfig,verbatim,enumitem,dcolumn,amssymb,appendix,physics,fmtcount,booktabs,csquotes,epsfig,times,geometry,array,mathrsfs,graphics,epstopdf,latexsym,comment,cancel}

\usepackage[normalem]{ulem}

\usepackage[mathscr]{euscript}
\def\Tr{\operatorname{Tr}}

\newtheorem{theorem}{Theorem}

\newtheorem{corollary}{Corollary}

\newtheorem*{remark*}{Remark}

\usepackage[colorlinks=true, citecolor=blue, urlcolor=blue, linkcolor = blue ]{hyperref}

\begin{document}
\author{Ritopriyo Pal\textsuperscript{1,2}}
\email{palritopriyo01@gmail.com}

\author{Priya Ghosh\textsuperscript{1}}
\email{priyaghosh@hri.res.in}

\author{Ahana Ghoshal\textsuperscript{3}}
\email{ahanaghoshal1995@gmail.com}
\author{Ujjwal Sen\textsuperscript{1}}
\email{ujjwal@hri.res.in}

\affiliation{\textsuperscript{1}Harish-Chandra Research Institute,  A CI of Homi Bhabha National Institute, Chhatnag Road, Jhunsi, Prayagraj  211 019, India\\ \textsuperscript{2}Indian Institute of Science Education and Research Kolkata, Mohanpur, 741246, West Bengal, India\\\textsuperscript{3}Naturwissenschaftlich-Technische Fakult\"{a}t, Universit\"{a}t Siegen, Walter-Flex-Stra\ss e 3, 57068 Siegen, Germany}

\begin{abstract}
Quantum multiparameter estimation offers a framework for the simultaneous estimation of multiple parameters, pertaining to possibly noncommutating observables.
While the optimal probe for estimating a single unitary phase is well understood — being a pure state that is an equal superposition of the eigenvectors of the encoding Hamiltonian corresponding to its maximum and minimum eigenvalues — the structure of optimal probes in the multiparameter setting remains more intricate. 
We investigate the simultaneous estimation of two phases, each encoded through arbitrary 
qubit Hamiltonians, using arbitrary weight matrices, and considering single-qubit probes. 
We also consider 
single-qutrit probes, 
for which the encoding Hamiltonians are chosen as $\mathbb{SU}(2)$ generators.
We find that in both the qubit and qutrit scenarios, the optimal probe is a coherent superposition of the eigenstates corresponding to the largest and smallest eigenvalues of the total encoding Hamiltonian, with a fixed - and not arbitrary - relative phase. Remarkably, this optimal probe is independent of the specific choice of weight matrix, making it a universally optimal input state for the estimation of any pair of $\mathbb{SU}(2)$ parameters, which can be re-parameterized to the phase estimation problem. Furthermore, we show that this probe also maximizes the determinant of the quantum Fisher information matrix, providing a handy tool for identifying the optimal probe state.
We also examine the role of commutativity between the generators of the unitary encodings. Our results demonstrate that a high degree of commutativity degrades the achievable precision, rendering estimation infeasible in the extreme case. However, maximal noncommutativity does not necessarily yield optimal precision. Notably, in the case of commuting generators and single-qutrit probes, the optimal input state must exhibit complete coherence in the common eigenbasis, though maximal coherence is not required.
\end{abstract}



\title{Role of phase of optimal probe in  noncommutativity vs coherence\\in quantum multiparameter estimation}



\maketitle

\section{introduction}
The precise estimation of unknown parameters in physical systems is crucial to both scientific research and technological advancement.
However, due to inherent quantum and classical limitations, perfect parameter estimation using finite resources is fundamentally unattainable. In such situations, the quantum parameter estimation theory---also known as quantum metrology~\cite{RevModPhys.90.035005, Gio-2006,doi:10.1126/science.1104149,RevModPhys.89.035002,RevModPhys.90.035006,Pirandola2018-et}---offers a powerful framework to determine the best achievable precision in estimating one or more parameters of interest.
Quantum systems have proven to be advantageous in several parameter estimation problems with the aid of their trademark non-classical resources~\cite{RevModPhys.91.025001}, such as entanglement~\cite{RevModPhys.81.865}  and quantum coherence~\cite{RevModPhys.89.041003}. 
In particular, quantum correlations---such as entanglement and squeezing~\cite{PhysRevD.23.1693}--- have been shown to surpass the classical shot-noise limit, enabling a quadratic gain in precision and achieving the Heisenberg scaling in the error corresponding to the estimation of the parameters of interest~\cite{Gio-2006, doi:10.1126/science.1104149,Wineland92,Gross_2012,nagata,Israel2014}.
These quantum advantages have led to significant theoretical and experimental developments in various domains, including gravitational wave detection~\cite{Aasi2013, Schnabel2010,Danilishin2012,RevModPhys.86.121}, quantum clocks \cite{RevModPhys.83.331,RevModPhys.87.637,Katori2011, Nichol2022}, quantum imaging \cite{PhysRevLett.102.253601,Lugiato2002, Dowling:15,PhysRevLett.109.123601, Genovese2016, Moreau2019}, biosensing~\cite{TAYLOR20161,Mauranyapin2017}, etc.  

When a single parameter encoded into a quantum probe is to be estimated within the framework of quantum statistical model, 
the mean squared error in the estimation of the parameter of interest is 
lower bounded by the quantum Cram\'er-Rao bound (QCRB)~\cite{PhysRevLett.72.3439}. This bound is obtained by optimizing over all possible estimators and quantum measurements for a given probe state.
Hence, finding a measurement and an estimator which can saturate the QCRB is an integral part of any quantum parameter estimation problem. 
In the case of single-parameter estimation, it is guaranteed that an optimal measurement exists which can saturate the bound~\cite{paris2009,PhysRevLett.72.3439,Helstrom1969}. Additionally, maximum-likelihood estimators (MLE) 
are known to asymptotically attain the QCRB in the limit of a large number of identically prepared probes~\cite{Kay, Lehmann1998-pk}. 
Therefore, once the optimal estimator and measurement are fixed, the remaining task is to optimize over the input quantum state to achieve the best possible precision. 
In particular, for the problem of estimation of a single phase $\phi$, encoded in a quantum system via a unitary of the form $U \coloneqq \exp(- i \phi H)$, the optimal probe is well known to be an equal superposition of the eigenvectors of the encoding Hamiltonian $H$, corresponding to the highest and lowest eigenvalues~\cite{Gio-2006,doi:10.1126/science.1104149}.

Extending quantum parameter estimation theory from single to multiple parameters~\cite{damptemp,animeshmultiphase,animeshmultidi,Vidrighin2014-yp,Rafalresolution,Szczykulska_2017,minimal-tradeoff,ALBARELLI2020126311,Liang_2018,Liu2015,su2,waveformPRL,waveformPRL2,Tsang_2013} is non-trivial due to the incompatibility relation between optimal measurements corresponding to the different parameters of interest. Despite this challenge, a plethora of different lower bounds on the mean squared error in the simultaneous estimation of  multiple parameters for any given probe have been developed~\cite{HELSTROM1967101,Helstrom1969,GillMassar,Holevo1982,HOLEVO1973337,Nagaoka,HiroshiNagaoka,PhysRevX.11.011028}.
However, not all of the bounds are achievable in all estimation scenarios. Among these bounds, the most widely studied and applied are the
multiparameter quantum Cram\'er-Rao bound and the Holevo Cram\'er-Rao bound (HCRB)~\cite{Holevo1982,HOLEVO1973337}. 
It is important to note that these multiparameter estimation bounds incorporate a positive-definite weight matrix in their formulation to obtain a scalar figure of merit. This weight matrix accounts for the trade-offs in the simultaneous estimation of different parameters of interest. 
As a result, the choice of the optimal probe state in multiparameter quantum estimation often depends on the particular weight matrix used.

In this work, we investigate the simultaneous estimation of two phases, each encoded via arbitrary $\mathbb{SU}(2)$ unitary operations, under arbitrary positive-definite weight matrices, using two classes of quantum probe states: (i) single-qubit and (ii) single-qutrit. 
Our primary objective is to identify optimal probe states, for both the cases, that yield the best achievable precision by minimizing appropriate quantum estimation bounds over all possible input states. 
In quantum multiparameter estimation, precision bounds such as the QCRB and the HCRB are derived for a fixed input state and depend on the chosen measurement strategy. However, in the single-qubit case, it is well established that the QCRB cannot, in general, be saturated for multiple unitarily encoded parameters~\cite{paris2022,Candeloro_2024,Carollo_2019}. As such, we employ the HCRB as our figure of merit in the single-qubit scenario. For single-qutrit probes, this issue does not arise, and the QCRB is a valid criterion for evaluating precision, allowing us to use it directly to determine the optimal probe.

We find that, in both cases, the optimal probe is a pure state that minimizes the chosen bound (HCRB or QCRB) and simultaneously maximizes the determinant of the quantum Fisher information matrix (QFIM). Importantly, the optimal probe is unique and independent of the weight matrix, establishing its universal optimality for the estimation of any pair of $\mathbb{SU}(2)$ encoded parameters— that can be re-parameterized to the estimation of the two phases.  Moreover, the optimal probe structure resembles that of the well-known single-parameter case: a coherent superposition of the eigenstates of the total encoding Hamiltonian corresponding to its maximum and minimum eigenvalues. However, unlike in single-parameter estimation, this superposition includes a fixed relative phase dependent on the specific values of the two parameters being estimated.
We also demonstrate that the degree of commutativity between the two encoding generators plays a key role in determining the precision of our estimation protocol. The estimation becomes increasingly unfeasible as the degree of commutativity between the two generators increases. However, maximal noncommutativity does not necessarily yield the highest estimation precision. Furthermore, when we consider a single-qutrit probe and two commuting encodings to be simultaneously estimated, we find that complete coherence in the eigenbasis of the total encoding Hamiltonians is a necessary condition, whereas the state need not be maximally coherent.

The rest of the paper is organized as follows. Section~\ref{not} introduces the notation used throughout the paper. In Sec.~\ref{multi}, we briefly discuss the fundamentals of quantum multiparameter estimation theory. Thereafter, in Sec.~\ref{qubit-strength} we optimize over the HCRB to find the optimal single-qubit probe in the two arbitrary unitary-encoded phase estimation problem for arbitrary weight matrix. This is followed by finding the optimal single-qutrit probe in the two $\mathbb{SU}(2)$ encoded phase estimation problem for arbitrary weight matrices using multiparameter QCRB in Sec.~\ref{qutrit-strength}. 
In Sec.~\ref{commuting-qutrit}, we will discuss the optimal probe for a two-phase estimation problem, where the parameters are encoded through $\mathbb{SU}(2)$ operations with commuting generators, using single-qutrit probes.
We conclude our discussion in Sec.~\ref{conclude}.

\section{Notations} 
\label{not}
Here, we introduce the notations used throughout the paper. Let \(\mathcal{B}(\mathcal{H})\) denote the set of positive semidefinite operators with unit trace on a Hilbert space \(\mathcal{H}\). A \(d\)-dimensional Hilbert space is denoted by \(\mathcal{H}^d\), and \(\mathbb{I}_d\) represents the identity operator on it. \(A^{-1}\) and \(A^T\) denote the inverse and transpose of a matrix \(A\), respectively. The notation \(A \succeq B\) indicates that \((A - B)\) is positive semidefinite, where \(A\) and \(B\) are matrices. Re$(\cdot)$ and Im$(\cdot)$ denote the real and imaginary parts of ``$\cdot$", respectively. $\Tr(\cdot)$ denotes trace of ``$\cdot$". The symbols ``$\cdot$" and ``$\times$" denote the dot and cross products between two vectors, respectively. $|\boldsymbol{v}|$ denotes the Euclidean norm of the vector $\boldsymbol{v}$. 
The \(x\)-, \(y\)-, and \(z\)-axes refer to the directions defined by the eigenbases of the Pauli-\(X\), Pauli-\(Y\), and Pauli-\(Z\) operators, respectively. The positive (negative) direction along each axis corresponds to the eigenstate associated with the higher (lower) eigenvalue of the respective Pauli operator.

\section{Quantum parameter estimation theory}
\label{multi}
Quantum parameter estimation theory seeks to determine unknown parameters encoded in quantum states with the highest possible precision, playing a central role in quantum metrology. Traditionally, the theory adopts a frequentist approach, where the parameter is treated as a fixed but unknown quantity, and the precision of an estimator is constrained by bounds such as the quantum Cram\'er-Rao bound, which depends on the quantum Fisher information of the probe state. However, an alternative perspective is offered by the Bayesian approach, where the parameter is modeled as a random variable with a known prior distribution~\cite{bayesfreq_phase_estimation}. 
In this work, we are more focused on the frequentist approach to quantum multiparameter estimation~\cite{HELSTROM1967101,Helstrom1969}, specifically examining the structure and implications of the multiparameter QCRB, the Holevo Cram\'er-Rao bound~\cite{HOLEVO1973337}, and the conditions under which these bounds are attainable.

Let $\bm{\varphi} \coloneqq (\varphi_1, \varphi_2, \dots, \varphi_N)^T \in \Theta \subset \mathbb{R}^N$ be a vector of $N$ unknown but approximately known fixed real parameters. These parameters are encoded into a finite-dimensional (say, $d$-dimensional) quantum state $\rho$ through a quantum process $\Lambda_{\bm{\varphi}}$, yielding the encoded state $\rho_{\bm{\varphi}}\coloneqq \Lambda_{\bm{\varphi}}\left[\rho\right]$. Here, the state $\rho$ is referred to as the input probe.
The goal of quantum multiparameter estimation is to simultaneously estimate all the $N$ parameters of such an encoded state, with the best possible precision, by performing appropriate measurements on the encoded state. A quantum measurement is described by a positive operator-valued measure (POVM), denoted by $\bm{\Pi} (\coloneqq \lbrace \Pi_{\omega}\succeq 0, \omega\in \Omega | \sum_{\omega\in\Omega}\Pi_{\omega}=\mathbb{I}_d\rbrace)$, where the set of possible outcomes of the measurement, denoted as $\Omega$, can be considered finite, without loss of generality \cite{PhysRevLett.98.190403}. The probability of obtaining the outcome $\omega$ upon measurement on the encoded state $\rho_{\bm{\varphi}}$ is given by $p(w;\bm{\varphi})(\coloneqq\Tr{(\Pi_{\omega}\rho_{\bm{\varphi}})})$, according to Born's rule. The values of all encoded parameters are inferred from the statistics of the measurement outcomes. To perform this inference, N functions—called estimators—are defined as $\tilde{\bm{\varphi}}\coloneqq\left(\tilde{\varphi_1},\tilde{\varphi}_2\cdots\tilde{\varphi}_N\right):\Omega\to\Theta$, where $\tilde{\varphi}_i$ denotes the estimator corresponds to $i$-th parameter of interest. 
The precision that is achieved in the simultaneous estimation of $N$ parameters is quantified via the mean squared error matrix (MSEM) of the $N$ estimators as follows:
\begin{align}   \bm{\Sigma}_{\rho_{\bm{\varphi}}}\left(\Pi,\bm{\tilde{\varphi}}\right)\coloneqq\sum_{\omega\in\Omega} p(\omega;\bm{\varphi})\left(\tilde{\bm{\varphi}}(\omega)-\bm{\varphi}\right)\left(\tilde{\bm{\varphi}}(\omega)-\bm{\varphi}\right)^T.
\end{align}
In the frequentist framework, estimators are considered to be~\textit{locally-unbiased}, satisfying the following two conditions:
\begin{align}
\label{local-unbiased}    \sum_{\omega\in\Omega}p\left(\omega;\bm{\varphi}^*\right)\bm{\tilde{\varphi}}(\omega)=\bm{\varphi}^*\,,\,\sum_{\omega\in\Omega} \tilde{\varphi}_i(\omega)\frac{\partial p(\omega;\bm{\varphi})}{\partial \varphi_j}\bigg|_{\bm{\varphi}=\bm{\varphi}^*}=\delta_{ij},
\end{align}
where $\bm{\varphi}^*$ denotes the vector of true values of the parameters of interest. The two conditions written in Eq.~\eqref{local-unbiased} essentially imply that the estimators are unbiased at the true value of the parameters of interest, $\bm{\varphi}^*$, as well as in a small neighborhood around $\bm{\varphi}^*$.   For locally-unbiased estimators, the MSEM is equal to the covariance matrix of the estimators for any choice of measurement $\bm{\Pi}$.
A multiparameter metrological bound helps to find the set of optimal estimators and measurements $\left(\{\bm{\tilde{\varphi}}^{\mathrm{opt}}\},\{\bm{\Pi}^{\mathrm{opt}}\}\right)$ that provides a lower achievable bound on the covariance matrix for a given encoded quantum state.

\subsection{Multiparameter quantum Cram\'er-Rao bound}
The multiparameter quantum Cram\'er-Rao bound (QCRB) provides a lower bound of the covariance matrix, and is given by
\begin{equation}
\label{multiqcrb} \bm{\Sigma}_{\rho_{\bm{\varphi}}}\left(\Pi,\bm{\tilde{\varphi}}\right)\succeq \mathcal{F}^{-1}_Q(\rho_{\bm{\varphi}}),
\end{equation}
 where $\mathcal{F}_Q$ is the quantum Fisher information matrix (QFIM)~\cite{Liu_2020}, whose elements are defined as
\begin{equation}  \left(\mathcal{F}_Q(\rho_{\bm{\varphi}})\right)_{ij}\coloneq\,\Re\left(\Tr{\left(\rho_{\bm{\varphi}}L_{\varphi_i}L_{\varphi_j}\right)}\right),
\end{equation}
where $i,j\in \{1,2,\ldots ,N\}$.
Here, $L_{\varphi_i}$ denotes the symmetric logarithmic derivative (SLD) operator~\cite{1054108} corresponding to the $i$-th parameter of interest $\varphi_i$, satisfying the condition 
\begin{equation*}
    \frac{\partial \rho_{\bm{\varphi}}}{\partial \varphi_i}=\frac{1}{2} \left(L_{\varphi_i}\rho_{\bm{\varphi}}+\rho_{\bm{\varphi}}L_{\varphi_i}\right).
\end{equation*}
Note that the QCRB depends solely on the encoded quantum state and is independent of both the choice of measurement and the estimator.
Moreover, if one wishes to estimate $\bm{\varphi}'=\bm{\varphi}'(\bm{\varphi})$ instead of $\bm{\varphi}$, then the quantum Fisher information matrices for the two different parametrizations are related to each other via a transformation using a Jacobian, as
\begin{align*}  \mathcal{F}_Q(\bm{\varphi})=\mathcal{J}^T\mathcal{F}_Q(\bm{\varphi}')\mathcal{J},
\end{align*}
where the elements of the Jacobian $\mathcal{J}$ are given by $\displaystyle{\mathcal{J}_{ij}\coloneqq\partial \varphi'_i/\partial \varphi_j}$ \cite{Liu_2020}.

Each diagonal element of the QFIM, $\left(\mathcal{F}_Q\right)_{ii}$, corresponds to the quantum Fisher information~\cite{Braunstein1995GeneralizedUR} of the parameter $\varphi_i$ and the inverse of $\left(\mathcal{F}_Q\right)_{ii}$ serves as the lower bound for the variance of the estimator $\tilde{\varphi}_i$. Hence, when $N=1$, the estimation problem reduces to the estimation of a single parameter. Note that the quantum Cram\'er-Rao bound~\cite{Gio-2006,Nagaoka} in this case is simply a scalar inequality rather than a matrix inequality written in Eq.~\eqref{multiqcrb}. 
In the quantum multiparameter estimation problem, a scalar figure of merit that is more convenient to handle than a matrix inequality is introduced. It is known as a cost or weight matrix $\mathcal{W}$, which determines the relative importance of accurate estimation of the different parameters and helps quantify their trade-offs. By definition, $\mathcal{W}$ is a real, symmetric, and positive-definite matrix. Using the weight matrix, a scalar bound for the multiparameter estimation problem can be formulated as follows:
\begin{align}
\label{cost}
\Tr{\left(\mathcal{W}\,\bm{\Sigma_{\rho_{\varphi}}\left(\Pi,\tilde{\varphi}\right)}\right)} \geq \mathcal{C}^{\textbf{SLD}},
\end{align}
where $\mathcal{C}^{\textbf{SLD}}\coloneqq \Tr{\left(\mathcal{W}\,\mathcal{F}_Q^{-1}(\rho_{\bm{\varphi}})\right)}$ is referred to as the SLD-based multiparameter QCRB.

For quantum single-parameter estimation problems, there always exists at least one measurement which saturates the QCRB. However, for $N\geq2$, the scalar bound corresponding to multiparameter estimation, $\mathcal{C}^{\textbf{SLD}}$, is generally unattainable because
it is not possible in general to find a single measurement that saturates the multiparameter QCRB due to the inherent noncommutativity of the SLDs associated with the parameters of interest. Therefore, several alternative achievable bounds are developed to find an achievable lower bound in multiparameter quantum estimation problem in the literature; one of such bounds is the Holevo Cram\'er-Rao bound or Holevo-Nagoka Cra\'mer-Rao bound~\cite{Holevo1982}.

\subsection{Holevo Cram\'er-Rao Bound}
Consider a vector of Hermitian operators $\bm{X} \coloneqq \left(X_1, X_2, \cdots X_N\right)^T$ where each $X_i$ is defined as
\begin{equation*}
X_i\coloneqq\sum_{\omega\in\Omega}\,\Pi_{\omega}\left(\tilde{\varphi}_i(\omega)-\varphi_i \right).
\end{equation*}
The condition of local-unbiasedness of estimators presented in Eq.~\eqref{local-unbiased} can be expressed easily in terms of $\bm{X}$ as
$\Tr{\left(\bm{\nabla}\rho_{\bm{\varphi}}\bm{X}^T\right)}=\mathbb{I}_N$ where $\bm{\nabla}\rho_{\bm{\varphi}}$ denotes the gradient of $\rho_{\bm{\varphi}}$ with respect to $\bm{\varphi}$. 
The Holevo Cram\'er-Rao bound (HCRB) is given by
\begin{align}
\label{holevo_defn}   &\Tr{\left(\mathcal{W}\,\bm{\Sigma}_{\rho_{\bm{\varphi}}}\left(\bm{\Pi},\bm{\tilde{\varphi}}\right)\right)} \geq \mathcal{\mathbf{C}}^{\textbf{Holevo}}(\rho_{\bm{\varphi}},\mathcal{W}),\nonumber
\\ &\text{where } \hspace{1 mm} \mathcal{\mathbf{C}}^{\textbf{Holevo}}(\rho_{\bm{\varphi}},\mathcal{W})\coloneqq\min_{\bm{X}}h_{\bm{\varphi}}(\bm{X}),\nonumber \\
&h_{\bm{\varphi}}(\bm{X})\coloneqq\big[\Tr{\left(\mathcal{W}\,\text{Re} Z[\bm{X}]\right)}+ \big|\big|\sqrt{\mathcal{W}}\,\text{Im} Z[\bm{X}]\sqrt{\mathcal{W}}\,\big|\big|_1\,;\nonumber\\ &\Tr{\left(\bm{\nabla}\rho_{\bm{\varphi}}\bm{X}^T\right)}=\mathbb{I}_N\big].
\end{align}
Here, $Z[\bm{X}] \coloneqq \Tr{\left(\rho_{\bm{\varphi}}\bm{X}\bm{X}^T\right)}$ is a positive-semidefinite matrix.
$||A||_1 (\coloneqq \Tr(\sqrt{A^\dagger A}))$ denotes the trace norm of operator $A$. In general, the minimization over $\bm{X}$ in Eq.~\eqref{holevo_defn} cannot be performed analytically. For finite-dimensional probe states, the HCRB can be evaluated using semidefinite programming~\cite{PhysRevLett.123.200503}. 

The Holevo Cram\'er-Rao bound is an achievable lower bound in quantum multiparameter estimation~\cite{10.1063/1.2988130,yang2019precision}, However, saturating the HCRB often requires collective measurements on many copies of the encoded state $\rho_{\bm{\varphi}}$, which may be impractical.
Nonetheless, in some cases – for example, when the encoded state is pure~\cite{Matsumoto_2002} or when estimating displacement parameters with Gaussian states
~\cite{Holevo1982}, the HCRB can be achieved via separable measurements.

\subsection{Achievability criterion of multiparameter QCRB}
The HCRB is bounded in relation to the multiparameter quantum Cram\'er-Rao bound as follows:
\begin{equation*}
    \mathcal{\mathbf{C}}^{\textbf{SLD}}\leq \mathcal{\mathbf{C}}^{\textbf{Holevo}}\leq 2\mathcal{\mathbf{C}}^{\textbf{SLD}}.
\end{equation*} 
Thus, HCRB is not only achievable but also tighter as compared to the multiparameter QCRB.
For positive-definite weight matrices—i.e., when all eigenvalues of $\mathcal{W}$ are strictly positive—there exists a necessary and sufficient condition under which HCRB coincides with multiparameter QCRB, making the latter achievable. This condition, known as the~\textit{weak-commutativity criterion} or the~\textit{criterion for asymptotic compatibility}~\cite{weak1,weak2}, states that the HCRB and multiparameter QCRB coincide if and only if the expectation values of the commutators of the SLDs corresponding to all pairs of parameters of interest vanish with respect to the encoded state $\rho_{\bm{\varphi}}$. That is, the multiparameter QCRB is achievable for positive-definite weight matrices when
\begin{equation}
\label{weak}
\text{Im}\left[\Tr\left(\rho_{\bm{\varphi}}L_{\varphi_i}L_{\varphi_j}\right)\right]=0;\quad\forall\,\varphi_i\neq \varphi_j.
\end{equation}

The weak-commutation condition can be expressed in matrix form as $D = 0$~\cite{Carollo_2019,paris2022}, where the matrix elements are given by \( D_{ij} \coloneqq \text{Im}\left[\Tr\left(\rho_{\bm{\varphi}} L_{\varphi_i} L_{\varphi_j}\right)\right] \). The matrix $D$ is known as the Uhlmann matrix.

We are now ready to present the results of our paper.

\section{Finding the optimal probe in estimation of two unitary encoding phases}
Quantum metrological bounds for parameter(s) estimation are typically derived by optimizing over measurements and estimators for a given input probe. To attain the highest possible precision, it is further necessary to optimize the bounds over the choice of input probe states. In this work, we aim to determine the maximum achievable precision in the simultaneous estimation of two phases encoded via an $\mathbb{SU}(2)$ process. Accordingly, we focus on finding the optimal probe in the simultaneous estimation of two phases encoded via an $\mathbb{SU}(2)$ process, considering single-copy probe acting on a two-dimensional Hilbert space in Subsection~\ref{qubit-strength}, followed by extending the analysis to the three-dimensional case in Subsection~\ref{qutrit-strength}. In the last Subsection, we will discuss properties of the optimal probe in two $\mathbb{SU}(2)$ encoded phase estimation problem for commuting generators using single-qutrit probes.

Throughout this work, we restrict our attention to uncorrelated probes.
It has been shown that for $N$ copies of uncorrelated probe, all metrological bounds considered in this paper—namely, the QCRB and HCRB—take the form \(\frac{\mathscr{C}}{N}\), where \(\mathscr{C}\) denotes the corresponding bound for a single-copy probe \cite{Matsumoto_2002}.
Therefore, we restrict our analysis on finding the optimal probe for the simultaneous estimation of two phases encoded through an $\mathbb{SU}(2)$ process, considering only the single-copy probe for the remainder of the paper.

\subsection{\texorpdfstring{$\mathbb{SU}(2)$}{SU(2)} encoding with single-qubit probes}
\label{qubit-strength}
In this subsection, we determine the optimal input probe for the simultaneous estimation of two phases encoded through an unitary process, when a single-copy probe acting on a two-dimensional Hilbert space are available. Although the QCRB is generally achievable when the weak commutation condition is satisfied, it cannot be used for the simultaneous estimation of two parameters encoded through a unitary process in a single qubit state~\cite{paris2022,Carollo_2019,Candeloro_2024}. Therefore, we employ HCRB to find the optimal probe in our scenario.
We begin by finding the optimal single-qubit pure probe for arbitrary unitary encoding processes and arbitrary weight matrices. Subsequently, we show that pure single-qubit probes exhibit better accuracy in the estimation of two unitary encoding phases than their mixed probe counterparts in our case.

Let us consider an arbitrary single-qubit pure state $\ket{\psi}$ as the input probe. In the Bloch sphere representation, the corresponding density matrix can be expressed as
\begin{equation*}
\rho=\ketbra{\psi}{\psi} = \frac{1}{2} (\mathbb{I}_2 + \bm{r} \cdot \bm{\sigma} ),
\end{equation*}
where $\bm{r} = (r_1, r_2, r_3)^T \in \mathbb{R}^3$ is the Bloch vector satisfying $r_1^2 + r_2^2 + r_3^2 = 1$. 
Let the two parameters to be estimated be denoted by $ \bm{\varphi} \coloneqq \{\varphi_1, \varphi_2\} \in \mathbb{R}$, which are encoded into the probe state $\ket{\psi}$ through an arbitrary unitary process $U(\bm{\varphi})$.
Let us consider that the unitary $U(\bm{\varphi})$ is generated by two Hermitian operators, $H_1$ and $H_2$, which are known $2 \times 2$ matrices and independent of $\bm{\varphi}$. The resulting encoded state, denoted by $\ket{\psi(\bm{\varphi})}$, is then given by
\begin{equation}
\label{unitary}
\ket{\psi(\bm{\varphi})} \coloneqq U(\bm{\varphi}) 
   \ket{\psi}=\exp{- i \left(\varphi_1 H_1+\varphi_2 H_2\right)}|\psi\rangle.
\end{equation}
The precision of the simultaneous estimation of $\varphi_1,\varphi_2$ is expected to depend on the  commutativity between $H_1$ and $H_2$.
Any general  $2 \times 2$ Hermitian matrix $A$ can be expressed as a linear combination of $\mathbb{I}_2$ along with the generators of $\mathbb{SU}(2)$ as
$A = \sum_{i=0}^{3} t_i \sigma_i$,
where  $\vec{\boldsymbol{t}} \in \mathbb{R}^4$ with $\vec{\boldsymbol{t}} \coloneqq \{t_0,\cdots, t_3\}$ and  $\sigma_i \in  \{ \mathbb{I}_2, \sigma_x, \sigma_y, \sigma_z \}$. Therefore, the two encoding Hamiltonians, $H_1$  and  $H_2$ can be written as follows:
\begin{equation*}
H_1 = \frac{1}{2}\bm{a}_1 \cdot \bm{\sigma}, \quad \text{and}\quad H_2 = \frac{1}{2}\bm{a}_2 \cdot \bm{\sigma}.
\end{equation*}
Here,  $\bm{a}_1$  and  $\bm{a}_2$  are unit vectors in  $\mathbb{R}^3$ , and $ \bm{\sigma} \coloneqq \left( \sigma_x, \sigma_y, \sigma_z \right)^T$ denotes the vector of Pauli matrices. We do not include the identity operator in $H1$ and $H_2$, since it only causes a relative shift in the energy levels between the two Hamiltonians and does not act as a non-trivial contributing factor while evaluating the precision in estimation.

Since the encoded state in our case is pure, SLDs that yield optimal precision in the quantum multiparameter estimation problem for each parameter $\varphi_i$ can be expressed as follows for the encoded state $\ket{\psi(\bm{\varphi})}$~\cite{1054108,Holevo1982}:
\begin{equation}
L_{\varphi_i} = 2\left(\ket{\partial_{\varphi_i} \psi(\bm{\varphi})}\bra{ \psi(\bm{\varphi})} + \ket{\psi(\bm{\varphi})} \bra{\partial_{\varphi_i}\psi(\bm{\varphi})}
\right),
\end{equation}
where $i\in\lbrace{1,2\rbrace}$, and $\ket{\partial_{\varphi_i} \psi(\bm{\varphi})} $ represents the first-order partial derivative of $\ket{\psi(\bm{\varphi})}$ with respect to $\varphi_i$.
For arbitrary finite-dimensional pure probes and arbitrary unitary encoding processes of two parameters, the elements of QFIM can be expressed in terms of Hamiltonians and probe states as~\cite{Liu2015}, 
\begin{align}
\label{cov}
&(\mathcal{F}_Q)_{ij} = 4\text{Cov}\left(\mathcal{H}_i,\mathcal{H}_j\right),\nonumber&&\\&\coloneqq 4\left[\text{Re}\left(\expval{\mathcal{H}_i\mathcal{H}_j}{\psi}\right)-\expval{\mathcal{H}_i}{\psi}\expval{\mathcal{H}_j}{\psi}\right],
\end{align}
where $\mathcal{H}_j \coloneqq i     (\partial_{\varphi_j} U^{\dagger}(\bm{\varphi}))U(\bm{\varphi}), \hspace{ 2 mm} j \in \lbrace 1, 2 \rbrace$. 
We shall call $\mathcal{H}_1 (\mathcal{H}_2)$ as the effective Hamiltonians corresponding to $H_1 (H_2)$ respectively. 
When $\lbrack H_1, H_2 \rbrack  = 0$, in the elements of QFIM written in Eq.~\eqref{cov},  $\mathcal{H}_i$ will be replaced by $(-H_i)$, for all $i \in \{1,2\}$.
In the general case, we shall use Wilcox's formula~\cite{10.1063/1.1705306} to find the explicit form of $\mathcal{H}_1$ and $\mathcal{H}_2$. Wilcox's formula provides a formula for differentiating the exponential of any operator, say $\hat{A}$, and is given by
\begin{equation}
\label{wilcox}
    \partial_l e^{\hat{A}}=\int_{0}^1 ds\, e^{s\hat{A}}\left(\partial_l \hat{A}\right)e^{(1-s)\hat{A}}.
\end{equation}
With the help of Wilcox's formula written in Eq.~\eqref{wilcox}, effective Hamiltonians can be expressed in the basis of the Pauli matrices, i.e., $\mathcal{H}_i=\frac{1}{2}\bm{\eta}_i\cdot\bm{\sigma}$,
with $i \in \{1,2\}$~\cite{su2}. The explicit forms of $\bm{\eta}_1$ and $\bm{\eta}_2$ are written in Appendix~\ref{eta_expr}. 

Thus, substituting the simplified forms of $\mathcal{H}_i$'s in Eq.~\eqref{cov}, the elements of the QFIM are given by
\begin{equation}  \left(\mathcal{F}_Q\right)_{\text{ij}}=\bm{\eta}_i\cdot\bm{\eta}_j-\left(\bm{r}\cdot\bm{\eta}_i\right)\left(\bm{r}\cdot\bm{\eta}_j\right),
\end{equation}
where $i,j\in\lbrace{1,2\rbrace}$. In this case, the elements of the Uhlmann matrix, $D$, are given by $D_{12}=-D_{21}=\bm{r}\cdot\left(\bm{\eta}_1\times\bm{\eta}_2\right)$.

The HCRB is a tighter bound than the QCRB, in  quantum  multiparameter estimation. For the QCRB to be achievable, the encoded state must satisfy the weak-commutation condition written in Eq.~\eqref{weak}. Specifically, in our scenario, the weak commutation condition takes the form $\bra{\psi (\bm{\varphi})} \left[L_{\varphi_1}, L_{\varphi_2}\right] \ket{\psi \left(\bm{\varphi}\right)} = 0$. This implies that the Bloch vector $\bm{r}$ of the input probe must lie in the plane spanned by the vectors $\bm{\eta}_1$ and $\bm{\eta}_2$, i.e., $(\bm{r}\cdot\left(\bm{\eta}_1\times\bm{\eta}_2\right)) = 0$, in order for the weak-commutation condition to be satisfied.
It is straightforward to see that when the Hamiltonians governing the unitary encoding process commute and the probe state is pure, the weak-commutation condition is automatically satisfied for any probe state \cite{Demkowicz-Dobrzański_2020, weak2}. As a result, QCRB becomes equal with HCRB for all pure probe states and no additional constraints are imposed on the choice of the probe while finding the optimal probe using the QCRB in such cases since, in such cases, QCRB is an automatically attainable bound.
However, imposing the above condition makes the QFIM singular for the noncommuting case when single-qubit probes are used. Therefore,  
the multiparameter QCRB cannot be a valid lower bound for the estimation of the two $\mathbb{SU}(2)$ phases using pure single-qubit probes which satisfy Eq.~\eqref{weak}, since the inverse of the QFIM does not exist. More generally, it was recently shown in Ref.~\cite{paris2022,Candeloro_2024,Carollo_2019} that whenever two parameters are encoded on an arbitrary single-qubit probe (pure or mixed) through a unitary and impose the weak-commutation condition, it ends up with a singular QFIM, irrespective of the
commutativity or noncommutativity of the Hamiltonians. This is a fact which is independent of re-parametrization, indicating that the multiparameter QCRB is not a useful bound in this estimation scenario. 
However, in such estimation scenarios (unitary encoding of noncommuting Hamiltonians), one can use the HCRB as our fundamental bound without imposing the weak-commutation condition. Thus, for a single-qubit probe, we will look for an input state which minimizes the HCRB, as it is still an achievable bound in this estimation task. 

Since varying the parameters $\bm{\varphi}$ of a unitary encoding $U(\bm{\varphi})$ does not vary the norm of the Bloch vector of the encoded single qubit state, our model is a \textit{D-invariant} quantum statistical model~\cite{e21070703}. For such models, the HCRB has been analytically derived~\cite{Suzuki} and it is given by
\begin{equation}
\label{Holevo_first}   
\mathcal{\mathbf{C}}^{ \textbf{\text{Holevo}}}=\Tr{(\mathcal{W}\mathcal{F}^{-1}_Q)}+\bigg|\bigg|\sqrt{\mathcal{W}}\mathcal{F}^{-1}_QD\mathcal{F}^{-1}_Q\sqrt{\mathcal{W}}\bigg|\bigg|_1.
\end{equation}

In order to find the expression of the HCRB for the two-parameter estimation problem explicitly, we need to choose a weight or cost matrix $\mathcal{W}$. We choose a general $2\times 2$ real, positive-definite weight matrix as follows:
\begin{equation}
\label{weight}
    \mathcal{W}=\begin{pmatrix}
        w_{11}\quad w_{12}\\
        w_{12}\quad w_{22}
    \end{pmatrix}
\end{equation}
with the conditions $w_{11},w_{22}>0$. The determinant of $\mathcal{W}$, given by $\det(\mathcal{W}) \coloneqq w_{11}w_{22} - w_{12}^2$, must also satisfy $\det(\mathcal{W}) > 0$ to ensure the positive definiteness of $\mathcal{W}$.
To ensure that both parameters are estimated, we restrict our analysis to full-rank weight matrices. This choice is motivated by the fact that a bound constructed using a rank-$K$ weight matrix effectively corresponds to the estimation of $K$ parameters~\cite{Belliardo_2021}.
\\
The determinant of the QFIM can be simplified as
\begin{align}
\label{simp_det}
\det(\mathcal{F}_Q)&= |\bm{\eta}_1|^2|\bm{\eta}_2|^2-(\bm{\eta}_1\cdot\bm{\eta}_2)^2-|\bm{\eta}_1|^2(\bm{r}\cdot\bm{\eta}_2)^2-\nonumber\\&|\bm{\eta}_2|^2(\bm{r}\cdot\bm{\eta}_1)^2 +2(\bm{\eta}_1\cdot\bm{\eta}_2)(\bm{r}\cdot\bm{\eta}_1)(\bm{r}\cdot\bm{\eta}_2)\nonumber,\\
&=\big|\bm{r}\big|^2\big|\bm{\eta}_1\times\bm{\eta}_2\big|^2-\big|\bm{r}\times\left(\bm{\eta}_1\times\bm{\eta}_2\right)\big|^2,\nonumber\\&=\left(\bm{r}\cdot\left(\bm{\eta}_1\times\bm{\eta}_2\right)\right)^2.
\end{align}
In the above simplification of $\det(\mathcal{F}_Q)$, we have used the facts that 
$\displaystyle{|\bm{a}_1|^2|\bm{a}_2|^2-(\bm{a}_1\cdot\bm{a}_2)^2} = |\bm{a}_1\times\bm{a}_2|^2$ for any two three-dimensional vectors $\bm{a}_1, \bm{a}_2$, the three-dimensional vector triple-product formula $\bm{r}\times\left(\bm{\eta}_1\times\bm{\eta}_2\right)=\bm{\eta}_1\left(\bm{r}\cdot\bm{\eta}_2\right)-\bm{\eta}_2\left(\bm{r}\cdot\bm{\eta}_1\right)$, and $\big|\bm{r}\big|=1$ since we only consider pure single-qubit probe states.

$\Tr{(\mathcal{W}\mathcal{F}^{-1}_Q)}$ can be simplified as
\begin{align}
&\Tr{(\mathcal{W}\mathcal{F}^{-1}_Q)}\nonumber\\&= \frac{1}{\det (\mathcal{F}_Q)}\bigg[w_{11}\left[|\bm{\eta}_2|^2-\left(\bm{r}\cdot\bm{\eta}_2\right)^2\right]+w_{22}\left[|\bm{\eta}_1|^2-\left(\bm{r}\cdot\bm{\eta}_1\right)^2\right]\nonumber \\&-2w_{12}\left[\,\bm{\eta}_1\cdot\bm{\eta}_2-(\bm{r}\cdot\bm{\eta}_1)(\bm{r}\cdot\bm{\eta}_2)\,\right] \bigg],\nonumber\\    &=\frac{1}{\left(\bm{r}\cdot\left(\bm{\eta}_1\times\bm{\eta}_2\right)\right)^2}\bigg[w_{11}\left(\big|\bm{r}\times\bm{\eta}_2\big|^2\right)+w_{22}\left(\big|\bm{r}\times\bm{\eta}_1\big|^2\right)\nonumber\\&-2w_{12}\left(\bm{r}\times\bm{\eta}_1\right)\cdot\left(\bm{r}\times\bm{\eta}_2\right)\bigg], \label{Q_expr}\\&= \frac{\mathcal{Q}}{\left(\bm{r}\cdot\left(\bm{\eta}_1\times\bm{\eta}_2\right)\right)^2}\label{hcrb_first},
\end{align}
 where $\mathcal{Q}\coloneqq w_{11}\left(\big|\bm{r}\times\bm{\eta}_2\big|^2\right)+w_{22}\left(\big|\bm{r}\times\bm{\eta}_1\big|^2\right)\nonumber-2w_{12}\left(\bm{r}\times\bm{\eta}_1\right)\cdot\left(\bm{r}\times\bm{\eta}_2\right),\nonumber$. The last equality in Eq.~\eqref{Q_expr} is derived by using the facts that $|\bm{r}|=1$, $|\bm{\eta}_i|^2-\left(\bm{r}\cdot\bm{\eta}_i\right)^2=|\bm{r}\times\bm{\eta}_i|^2,\,\text{ where }i\in\lbrace{1,2\rbrace}$ and finally, $\left(\bm{r}\times\bm{\eta}_1\right)\cdot\left(\bm{r}\times\bm{\eta}_2\right)=\left(\bm{\eta}_1\cdot\bm{\eta}_2\right)-\left(\bm{r}\cdot\bm{\eta}_1\right)\left(\bm{r}\cdot\bm{\eta}_2\right).$

The Uhlmann matrix $D$ is an anti-symmetric matrix, whereas $\mathcal{F}_Q$ is a real symmetric, positive semidefinite matrix. 
Using the symmetry properties of $D$ and $\mathcal{F}_Q$, in the two-parameter problem, it is simple to obtain
\begin{equation*}
    \mathcal{F}^{-1}_Q D \mathcal{F}^{-1}_Q=\frac{D}{\det  (\mathcal{F}_Q)}.
\end{equation*}
Similarly, we also have
\begin{equation*}
\sqrt{\mathcal{W}}\mathcal{F}^{-1}_Q D \mathcal{F}^{-1}_Q\sqrt{\mathcal{W}}=\frac{\sqrt{\det\left(\mathcal{W}\right)}}{\det\left(\mathcal{F}_Q\right)}D.
\end{equation*}
Therefore, the trace-norm term of the HCRB in our case can be written as
\begin{align}
\bigg|\bigg|\sqrt{W} \mathcal{F}^{-1}_Q D\mathcal{F}^{-1}_Q \sqrt{W}\bigg|\bigg|_1&=2\sqrt{\det(\mathcal{W})}\frac{|\bm{r}\cdot\left(\bm{\eta}_1\times\bm{\eta}_2\right)|}{\det\left(\mathcal{F}_Q\right)},\nonumber\\
&=\frac{2\sqrt{\det(\mathcal{W})}}{\big|\bm{r}\cdot\left(\bm{\eta}_1\times\bm{\eta}_2\right)\big|} \label{hcrb_second}.
\end{align}

Thus, substituting the forms of $\Tr{(\mathcal{W}\mathcal{F}^{-1}_Q)}$, and $\bigg|\bigg|\sqrt{W} \mathcal{F}^{-1}_Q D\mathcal{F}^{-1}_Q \sqrt{W}\bigg|\bigg|_1$, written in Eqs.~\eqref{hcrb_first}, and~\eqref{hcrb_second} respectively, into Eq.~\eqref{Holevo_first}, the HCRB for input probe $\ket{\psi}$ in our case can be written as
\begin{align}
\label{HCRB}
&\mathcal{\mathbf{C}}^{\textbf{\text{Holevo}}}=\frac{\mathcal{Q}}{\left(\bm{r}\cdot\left(\bm{\eta}_1\times\bm{\eta}_2\right)\right)^2}+\frac{2\sqrt{\det(\mathcal{W})}}{\big|\bm{r}\cdot\left(\bm{\eta}_1\times\bm{\eta}_2\right)\big|},
\end{align}
where 
$\boldsymbol{r}$ denotes the Bloch vector of the input probe $\ket{\psi}$.



\begin{theorem}
In the estimation of two arbitrary phases, encoded into any single-qubit probe through any arbitrary unitary, the optimal single-qubit probe that exhibits the best achievable precision in the estimation by minimizing the Holevo Cram\'er-Rao bound is unique and 
its Bloch vector is given by
\begin{equation}
\label{best_probe}    \bm{r}_{\text{opt}}=\pm\frac{\bm{\eta}_1\times\bm{\eta}_2}{\big|\bm{\eta}_1\times\bm{\eta}_2\big|}.
\end{equation}
Moreover, the optimal probe is independent of the choice of the weight/cost matrix and also, this optimal probe maximizes the determinant of the quantum Fisher information matrix.
\end{theorem}

\begin{proof}
We present the proof in two steps. First, we find the optimal input state among all pure single-qubit input probes in this estimation task. Then, we show that no mixed single-qubit input state can achieve better precision than this optimal pure state in this scenario.

Let us consider an orthonormal basis of $\mathbb{R}^3$ denoted by $\lbrace{\bm{e}_i\rbrace};\text{ where }i\in\lbrace{1,2,3\rbrace}$, such that $\bm{\eta}_1$ points in the direction of $\bm{e}_1$. We also consider $\bm{\eta}_2$ to lie in the plane spanned by $\bm{e}_1$ and $\bm{e}_2$, making an angle $\gamma$ with $\bm{\eta}_1$. Thus, $(\bm{\eta}_1\times\bm{\eta}_2)$ points in the direction of $\bm{e}_3$.
We can always consider the Bloch vector $\bm{r}$ of any single qubit pure state to have the following components along the three orthogonal directions corresponding to this basis as
\begin{equation*} \bm{r}\cdot\bm{e}_1=\sin{(\chi)}\cos{(\varsigma)}\,,\,\bm{r}\cdot\bm{e}_2=\sin(\chi)\sin{(\varsigma)}\,,\,\bm{r}\cdot\bm{e}_3=\cos{(\chi)},
\end{equation*}
 where $\chi\in\left[0,\pi\right]$ and $\varsigma\in[0,2\pi]$. Substituting the expression of $\bm{r}$ written in the $\lbrace{\bm{e}_i\rbrace}$ basis in the expression of $\mathcal{Q}$, given in Eq.~\eqref{Q_expr}, we obtain
\begin{align*}
\mathcal{Q}&=t_1\cos^2(\chi)+t_2\sin^2(\chi),
\end{align*}
where $t_1\coloneqq \bigg[w_{22}|\bm{\eta}_1|^2+w_{11}|\bm{\eta}_2|^2-2w_{12}|\bm{\eta}_1||\bm{\eta}_2|\cos{(\gamma)}\bigg],$ and $
t_2\coloneqq \bigg[w_{11}|\bm{\eta}_2|^2\sin^2{(\gamma-\varsigma)}+w_{22}|\bm{\eta}_1|^2\sin^2{(\varsigma)}+2|\bm{\eta}_1||\bm{\eta}_2|w_{12}\sin{(\gamma-\varsigma)}\sin{(\varsigma)}\bigg]$.
Let us define two two-dimensional vectors $\bm{v}_1,\bm{v}_2$ as $\bm{v}_1\coloneqq
\begin{pmatrix}   |\bm{\eta}_1|,\hspace{0.05mm}|\bm{\eta}_2|
\end{pmatrix}^T$
and $\bm{v}_2\coloneqq
\begin{pmatrix} |\bm{\eta}|_1\sin{(\varsigma)},&\hspace{0.05mm}|\bm{\eta}_2|\sin{(\gamma-\varsigma)}
\end{pmatrix}^T$ respectively.
The terms $t_1$ , $t_2$ can be expressed in terms of vectors $\bm{v}_1,\bm{v}_2$ as follows:
\begin{align*}
&t_1=\bm{v}^{T}_1
\begin{pmatrix}
        w_{22}&-w_{12}\cos{(\gamma)}\\
        -w_{12}\cos{(\gamma)}&w_{11}
\end{pmatrix}\bm{v}_1,\\
&\text{and} \quad 
t_2=\bm{v}^T_2
\begin{pmatrix}
     w_{22}&w_{12}\\
        w_{12}&w_{11}
\end{pmatrix}
\bm{v}_2.
\end{align*}
Note that both the terms $t_1$ and $t_2$ are non-negative because $\mathcal{W}$ is a positive-definite matrix.
Similarly, using the expression of the Bloch vector $\bm{r}$ in the $\lbrace{\bm{e}_i\rbrace}$ basis, the determinant of the QFIM is given by
$\det(\mathcal{F}_Q)=\left(\bm{r}\cdot\left(\bm{\eta}_1\times\bm{\eta}_2\right)\right)^2= |\bm{\eta}_1|^2|\bm{\eta}_2|^2\sin^2{(\gamma)}\cos^2{(\chi)}$.

The first term of the HCRB, written in Eq.~\eqref{hcrb_second}, can be simplified as
\begin{align*}
\frac{\mathcal{Q}}{\left(\bm{r}\cdot\left(\bm{\eta}_1\times\bm{\eta}_2\right)\right)^2}=\frac{1}{|\bm{\eta}_1|^2|\bm{\eta}_2|^2\sin^2(\gamma)}\left[t_1+t_2\tan^2(\chi)\right].
\end{align*}
Since $t_1$ and $t_2$ are non-negative, the minimum of $\frac{\mathcal{Q}}{\left(\bm{r}\cdot\left(\bm{\eta}_1\times\bm{\eta}_2\right)\right)^2}$ over all possible $\chi$ occurs at $\chi=0$, i.e., when $\bm{r}$ is in the direction of $(\bm{\eta}_1\times\bm{\eta}_2)$.
We can see that the second term of the HCRB written in Eq.~\eqref{hcrb_second} will also get minimum over all possible $\bm{r}$ when $\bm{r}$ is aligned with the direction of $(\bm{\eta}_1\times\bm{\eta}_2)$.
Therefore, the probe state mentioned in Eq.~\eqref{best_probe} is the optimal probe among all pure single-qubit input states in this case.

So far, we have found the optimal probe state for estimating the two unitarily encoded parameters $\bm{\varphi}$ by minimizing the HCRB over all pure single-qubit probes. We now show that no mixed single-qubit probe can achieve better precision than this optimal pure state given in Eq.~\eqref{best_probe} as follows.

Let us consider a mixed single-qubit input state, denoted by $\rho$. 
It is known that when the encoded process is a unitary parametrization process and input states are single-qubit mixed states, the QFIM of the encoded state $\rho_{\bm{\varphi}} (\coloneqq U(\bm{\varphi}) \rho U(\bm{\varphi})^\dagger)$ corresponding to the input state $\rho$, in such a scenario, can be expressed as~\cite{Liu2015}
\begin{align*}   \mathcal{F}_Q(\rho_{\bm{\varphi}})=\left(2\Tr(\rho^2)-1\right)\,\mathcal{F}_Q\left(\ketbra{\psi(\bm{\varphi})}{\psi(\bm{\varphi})}\right).
\end{align*}
Here $\mathcal{F}_Q\left(\ketbra{\psi(\bm{\varphi})}\right)$ denotes the QFIM for any one of the two eigenstates of the encoded state $\rho(\bm{\varphi})$.
The Uhlmann matrix for the encoded state $\rho(\bm{\varphi})$ can be written as
\begin{align*}
D(\rho_{\bm{\varphi}})=\pm\left(2\Tr(\rho^2)-1\right)^{3/2}\,D\left(\ketbra{\psi(\bm{\varphi})}{\psi(\bm{\varphi})}\right),
\end{align*}
where $``+"$ sign holds when $\ket{\psi(\bm{\varphi})}$ is the eigenstate of the state $\rho(\bm{\varphi})$ with the greater eigenvalue; otherwise, $``-"$ sign holds. See Appendix~\ref{purity-uhlmann} for a detailed derivation of $D(\rho_{\bm{\varphi}})$.

The HCRB will be for the encoded state $\rho_{\bm{\varphi}}$  corresponding to the input state $\rho$ will have the form
\begin{widetext}
  \begin{align*} 
&\mathcal{\mathbf{C}}^{ \textbf{\text{Holevo}}} (\rho_{\bm{\varphi}}) = \frac{1}{\left(2\Tr(\rho^2)-1\right)}\Tr{\bigg[\mathcal{W}\mathcal{F}^{-1}_Q\big(\ket{\psi(\bm{\varphi})}\big) \bigg]}+ \frac{1}{\sqrt{\left(2\Tr(\rho^2)-1\right)}}
\bigg|\bigg|\sqrt{\mathcal{W}}\mathcal{F}^{-1}_Q \big(\ket{\psi(\bm{\varphi})}\big) D \big(\ket{\psi(\bm{\varphi})}\big) \mathcal{F}^{-1}_Q \big(\ket{\psi(\bm{\varphi})}\big)\sqrt{\mathcal{W}}\bigg|\bigg|_1, 
\end{align*}  
\end{widetext}
where $\ket{\psi(\bm{\varphi})}$ is one of the two eigenstates of the encoded state $\rho_{\bm{\varphi}}$. Therefore, $\mathcal{\mathbf{C}}^{ \textbf{\text{Holevo}}}$ of state $\ket{\psi(\bm{\varphi})}$, must be lower than $\mathcal{\mathbf{C}}^{ \textbf{\text{Holevo}}}$ of state $\rho_{\bm{\varphi}}$ where $\rho_{\bm{\varphi}} \coloneqq U(\bm{\varphi}) \rho U(\bm{\varphi})^\dagger, \ket{\psi(\bm{\varphi})} \coloneqq U(\bm{\varphi}) \ket{\psi}$. 

Since there always exists a pure input probe for each mixed probe, which provides better precision compared to the mixed input probe in this scenario, and the state written in Eq.~\eqref{best_probe} is the optimal probe among all pure input states in this case, therefore the state written in Eq.~\eqref{best_probe} is the optimal one among all input states in this estimation task when we use HCRB as our metrological bound.

Moreover, note that the determinant of the QFIM of any pure single-qubit input state in this case is given by
$\det(\mathcal{F}_Q)=16|\bm{\eta}_1|^2|\bm{\eta}_2|^2\sin^2{(\gamma)}\cos^2{(\chi)}$
and we can see that it gets maximum among all pure single-qubit input states when the Bloch vector of the input state has the form written in Eq.~\eqref{best_probe}. On the other hand, for any mixed single-qubit input state $\rho$, the determinant of the QFIM will have the form $\det(\mathcal{F}_Q (\rho_{\bm{\varphi}})) = \left(2\Tr(\rho^2)-1\right)^{2} \det(\mathcal{F}_Q (\ket{\psi_{\bm{\varphi}}}))$, where $\rho_{\bm{\varphi}}$ is the encoded state corresponding to the input state $\rho$ in this case and $\ket{\psi_{\bm{\varphi}}}$ is one of the eigenstates of $\rho_{\bm{\varphi}}$. Since $\Tr{(\rho^2)} \leq 1$, $\det(\mathcal{F}_Q (\rho_{\bm{\varphi}})) \leq \det(\mathcal{F}_Q (\ket{\psi_{\bm{\varphi}}}))$, i.e., for each mixed single-qubit input state, the determinant of the QFIM corresponding to the mixed state is always less than the determinant of the QFIM of the encoded state corresponding to one of the eigenstates of the mixed state. Hence, the input state, written in Eq.~\eqref{best_probe}, will provide the maximum determinant of QFIM among all single-qubit input states in this estimation task.

Thus, it concludes the proof. 
\end{proof}

Substituting $\bm{r}=\bm{r}_{opt}$, as given in Eq.~\eqref{best_probe}, in the expression of the HCRB given in Eq.~\eqref{min_holevo}, we obtain the minimum HCRB in the simultaneous estimation of two arbitrary phases via arbitrary unitary encoding processes using all single-qubit probes, as discussed below.

\begin{corollary}
The optimal or minimum possible value of the Holevo Cram\'er-Rao bound for the simultaneous estimation of any two arbitrary phases $\bm{\varphi} \coloneqq \{\varphi_1, \varphi_2\} $, encoded through unitary transformation using single-qubit probes, is given by
\begin{align}
\label{min_holevo}
&\min_{\displaystyle{\rho\in \mathcal{H}^2}}\mathbf{C^{\,Holevo}}\Big(U(\bm{\varphi})\rho U^{\dagger}(\bm{\varphi});\mathcal{W}\Big)\nonumber\\&=\frac{w_{11}\big|\bm{\eta}_{2}\big|^2+w_{22}\big|\bm{\eta}_1\big|^2-2w_{12}\left(\bm{\eta}_1\cdot\bm{\eta}_2\right)}{\big|\bm{\eta}_1\times\bm{\eta}_2\big|^2}+\frac{2\sqrt{\det\left(\mathcal{W}\right)}}{\big|\bm{\eta}_1\times\bm{\eta}_2\big|},
\end{align}
where the forms of $\bm{\eta}_1$ and $\bm{\eta}_2$ are written in Appendix~\ref{eta_expr}. $\mathcal{W}$ denotes the weight matrix, parameterized by $w_{11},w_{22},~\text{and} \hspace{2 mm} w_{12}$, and it is written in Eq.~\eqref{weight}. 
\end{corollary}

Note that the single-qubit probe written in Eq.~\eqref{best_probe}, which acts as the optimal single-qubit probe in the simultaneous estimation of the two arbitrary $\mathbb{SU}(2)$ phases $\bm{\varphi}$, is independent of the choice of weight matrix.  Since using different weight matrices is equivalent to the estimation of a re-parameterized set of parameters $\bm{\varphi}'=\bm{\varphi}'(\bm{\varphi})$~\cite{paris2022}, we can conclude that the optimal single-qubit probe is indeed optimal for the simultaneous estimation of any arbitrary two $\mathbb{SU}(2)$ encoded parameters which can be re-parameterized into our original phase-estimation problem.

\begin{corollary}
The optimal single-qubit probe state in the simultaneous estimation of two arbitrary unitary encoded phases is an equally weighted superposition of the highest and lowest eigenstates of the total encoding Hamiltonian with a fixed relative phase, i.e., it is maximally coherent in the eigenbasis of the total encoding Hamiltonian. 
Thus, maximal quantum coherence in the eigenbasis of the total encoding Hamiltonian is a necessary condition for achieving the best possible precision in the simultaneous estimation of two arbitrary unitary encoded phases within the single-qubit probe, but it is 
not sufficient. 
\end{corollary}

\begin{proof}
For the sake of simplicity, let us choose the axes of encoding Hamiltonians to be $\bm{a}_1 \coloneqq \left(1,0,0\right)^T$ and $\bm{a}_2 \coloneqq \left(\cos(\theta),\sin(\theta),0\right)^T$ with $\theta \in \left[0,\pi\right]$.     
Our encoding unitary, $U(\bm{\varphi}) \coloneqq \exp{-\frac{i}{2}(\varphi_1 \bm{a}_1\cdot \bm{\sigma}+\varphi_2 \bm{a}_2\cdot\bm{\sigma})}$,
is a rotation operation about an axis whose unit vector is given by
\begin{equation}
\label{coherent_normal}   \hat{n}_{rot}=\frac{\left(\varphi_1+\varphi_2\cos{(\theta)},\varphi_2\sin{(\theta)},0\right)}{\sqrt{\varphi^2_1+\varphi^2_2+2\varphi_1\varphi_2\cos{(\theta)}}}.
\end{equation}

Let us consider pure single-qubit states which are equal superpositions of the two eigenstates of total encoding Hamiltonian $\displaystyle{H(\bm{\varphi})=\frac{1}{2}\left(\varphi_1\,\bm{a}_1\cdot\bm{\sigma}+\varphi_2\,\bm{a}_2\cdot\bm{\sigma}\right)}$, 
and are perpendicular to $\hat{n}_{rot}$. Let us denote the Bloch vectors of these pure single-qubit pure states by $\bm{r}=(r_x,r_y,r_z)^T$ with $|\bm{r}|=1$. Then, these single-qubit states must satisfy the following condition: 
\begin{equation}
\bm{r}\cdot\hat{n}_{rot}=r_x(\varphi_1+\varphi_2\cos{(\theta)})+r_y\varphi_2\sin{(\theta)}=0.
\end{equation}
Substituting the expressions of $\bm{a}_1$ and $\bm{a}_2$ into $\bm{\eta}_1,\bm{\eta}_2$ given in the Appendix~\ref{eta_expr}, it can be easily found that $\left(\bm{\eta}_1\times\bm{\eta}_2\right)\cdot\hat{n}_{rot}=0$.

Thus, it completes the proof.
\end{proof}

\begin{figure}[h!]
    \centering
   \includegraphics[width=0.75\linewidth]{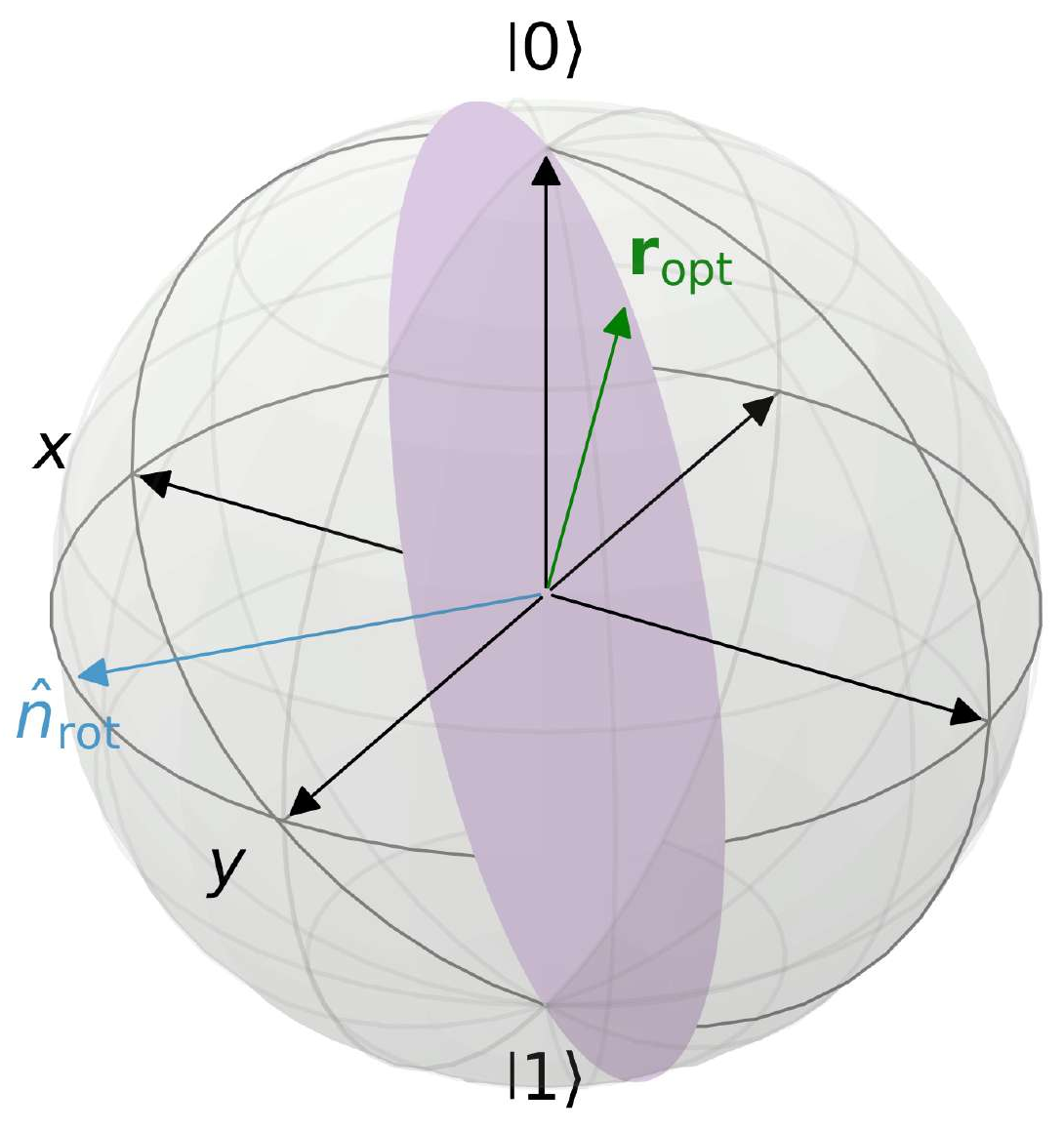}
    \caption{\textbf{Pictorial representation of the optimal single-qubit probe in the simultaneous estimation of two arbitrary unitary encoding phases}. 
    The grey sphere denotes the Bloch sphere of single-qubit states. The black arrows represent the orthogonal directions $x,y,z$ with $\ket{0}$ and $\ket{1}$ representing the two orthonormal eigenstates of $\sigma_z$. Here the blue arrow with the annotation $\hat{n}_{rot}$ represents the unit vector along the axis about which the encoding unitary operation rotates any single-qubit state. 
    The purple disc represents the equatorial plane perpendicular to $\hat{n}{\text{rot}}$ and passing through the center of the sphere. 
    The green arrow, labelled $\bm{r}{\text{opt}}$, denotes the Bloch vector of the optimal probe state that minimizes the HCRB in the simultaneous estimation of two arbitrary unitary encoded phases. Since the optimal state is always pure, regardless of the choice of weight matrix or encoding Hamiltonians, $\bm{r}{\text{opt}}$ touches the surface of the Bloch sphere. Importantly, $\bm{r}_{\text{opt}}$ lies in the purple disc, corresponding to an equal superposition of the eigenstates of the total encoding Hamiltonian.}
    \label{fig:max_coh}
\end{figure}
Note that similar to the single unitary phase estimation problem, the optimal single-qubit probe in the estimation of two $\mathbb{SU}(2)$ phases is an equally weighted superposition of the eigenstates of the total encoding Hamiltonian. However,
unlike the single unitary phase estimation problem using single-qubit probes, here maximum coherence is a necessary condition but not sufficient because the optimal probe in this case has a fixed relative phase in the equal superposition of the eigenstates of the total encoding Hamiltonian.

Now, we consider the encoding axis $\bm{a}_1$ (associated with Hamiltonian $H_1$) to be aligned along the positive $x$-direction, while $\bm{a}_2$ (associated with $H_2$) lies in the $xy$-plane and forms an angle $\theta$ with $\bm{a}_1$. The weight matrix is chosen with elements $\displaystyle w_{11} = w_{22} = 1.0$ and $\displaystyle w_{12} = 0.2$.
We investigate the behavior of the optimal HCRB in the simultaneous estimation of two arbitrary $\mathbb{SU}(2)$-encoded phases as a function of the angle $\theta$ between the encoding axes corresponding to $H_1$ and $H_2$, for different values of the parameters of interest, when the probe is restricted to a two-dimensional Hilbert space. This behavior is illustrated in Fig.~\ref{holevo_plot}.
We observe that when the angle $\theta$ between the encoding axes is close to $0$ or $\pi$, i.e, $H_1$ and $H_2$ are highly commuting, the optimal value of the HCRB is quite high, thereby indicating a poor precision in the simultaneous estimation of our parameters of interest.  The optimal precision gets better as $H_1$ and $H_2$ become more noncommuting, i.e., as $\theta$ approaches $\frac{\pi}{2}$. To highlight the behaviour of the optimal HCRB values in greater detail, we include an inset in Fig.(\ref{holevo_plot}), which shows a zoomed-in view of the curves around their minima. The inset restricts the range of the $y$-axis, allowing for a clearer comparison of the position and nature of the minima across the different curves. The black dotted line in the inset marks the position of $\theta=\pi/2$. We observe that the best precision does not necessarily occur when the two axes $\bm{a}_1$ and $\bm{a}_2$ are orthogonal. Although we depict this here for a particular choice of weight matrix, we have observed similar features for other weight matrices as well.

Moreover, it is interesting to explore estimation scenarios in which quantum metrology involves the estimation of extremely weak or feeble fields. In this context, we now present an approximate yet simplified expression for the Bloch vector of the optimal probe state in our estimation task, assuming the parameters of interest are very small, i.e., $\varphi_1, \varphi_2 \ll 1$.

\subsubsection{Small parameter approximation}
In the regime when the parameters of interest in the two-phase estimation task using single-qubit probes are very small, i.e., $\varphi_1,\varphi_2<<1$, the components of the Bloch vector of optimal probe which exhibits best precision from the perspective of HCRB are approximately given by
\begin{align*}
    r^x_{\text{opt}}\approx \mp \frac{\varphi_2\sin{(\theta)}}{\sqrt{4+\varphi^2_1+\varphi^2_2+2\varphi_1\varphi_2\cos{(\theta)}}}\to 0,\\
    r^y_{\text{opt}}\approx\pm \frac{\varphi_1+\varphi_2\cos{(\theta)}}{\sqrt{4+\varphi^2_1+\varphi^2_2+2\varphi_1\varphi_2\cos{(\theta)}}}\to 0,\\
    r^z_{\text{opt}}\approx\pm \frac{2}{\sqrt{4+\varphi^2_1+\varphi^2_2+2\varphi_1\varphi_2\cos{(\theta)}}}\to \pm 1.
\end{align*}

Since the optimal probe for the simultaneous estimation of two moderately valued $\mathbb{SU}(2)$-encoded phases depends on the vectors $\bm{\eta}_1$ and $\bm{\eta}_2$, which themselves are functions of the parameters of interest, some prior knowledge of the approximate values of parameters of interest is generally required before the estimation process. This necessity also arises because the measurements required to saturate the HCRB typically depend on the true values of parameters of interest. However, note that in the limit where the parameters are very small, this dependence is significantly reduced, and prior knowledge of their approximate values of parameters of interest is no longer needed to determine the optimal probe state in this two-phase estimation scenario.

\begin{figure}[h!]
    \centering
    \includegraphics[width=1.0\linewidth]{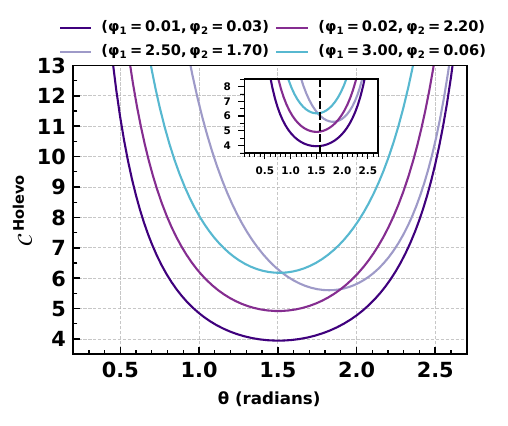}
    \caption{\textbf{How the noncommutativity of $H_1$ and $H_2$ affects the optimal HCRB for single-qubit probes.} We study the behavior of the optimal HCRB in the simultaneous estimation of two arbitrary $\mathbb{SU}(2)$-encoded phases as a function of the angle $\theta$ between the encoding axes corresponding to $H_1$ and $H_2$, for different values of phases, when the probe is restricted to a two-dimensional Hilbert space.  
    We align the encoding axis $\bm{a}_1$ (corresponding to Hamiltonian $H_1$) along the positive $x$-direction, while $\bm{a}_2$ (corresponding to Hamiltonian $H_2$) lies in the $xy$-plane, forming an angle $\theta$ with $\bm{a}_1$. The chosen weight matrix has elements $\displaystyle w_{11} = w_{22} = 1.0$ and $\displaystyle w_{12} = 0.2$.
    When the two encoding axes are almost parallel or anti-parallel to each other, we see that the error in estimation is quite high in all four two-phase estimation tasks. As the noncommutativity of the two encodings increases, the precision in estimation becomes better. The inset shows a zoomed-in image of the minima corresponding to the different curves of optimal HCRB, with the black dotted line denoting the position of $\theta=\pi/2$, i.e., when the encoding axes are orthogonal to each other. The inset shows that orthogonal encodings do not correspond to the best precision in the two-phase simultaneous estimation tasks.
    The vertical axis is dimensionless, while the horizontal axis has a unit of radians.
    }
    \label{holevo_plot}
\end{figure}

\subsection{\texorpdfstring{$\mathbb{SU}(2)$}{SU(2)} encoding single-qutrit probes}
\label{qutrit-strength}
In this subsection, we try to find the optimal probe in the estimation of two $\mathbb{SU}(2)$ phases when probes are restricted to three-dimensional Hilbert space.
Let us consider a two-parameter $\mathbb{SU}(2)$ encoding of the form
\begin{align*}
U(\bm{\varphi})&=\exp{-i\left(\varphi_1 H_1 +\varphi_2 H_2 \right)} \\&\coloneqq \exp{-i\left(\varphi_1\,\bm{a}_1\cdot\bm{J}+\varphi_2\,\bm{a}_2\cdot\bm{J}\right)},
\end{align*}
where $\bm{J}=(J_x,J_y,J_z)^T$ denotes the $\mathbb{SU}(2)$ generators in their spin-$1$ representation and $\bm{a}_1$, $\bm{a}_2$ are any two arbitrary three-dimensional vectors. Our aim is to find the probe that provides the best precision in the simultaneous estimation of $\varphi_1\,,\varphi_2$.
Let's say that we encode the phases $\varphi_1$, $\varphi_2$, which we want to estimate simultaneously, into a pure single-qutrit probe, denoted by $\ket{\psi_0}$. As the previous single-qubit probe case, the effective Hamiltonians $\mathcal{H}_1,\mathcal{H}_2$, corresponding to $H_1$ and $H_2$ respectively, are defined here using $\bm{\eta}_1$ and $\bm{\eta}_2$ respectively. Since the commutation relations of $\mathbb{SU}(2)$ generators are independent of the representation we work with, the expressions for $\bm{\eta}_1$ and $\bm{\eta}_2$ are identical to the ones in the single-qubit case. Hence, the elements of the QFIM in this case are given by
\begin{align}   \left(\mathcal{F}_Q\right)_{ij}=4\text{Re}\left[\expval{\left(\bm{\eta}_i\cdot\bm{J}\right)\left(\bm{\eta}_j\cdot\bm{J}\right)}{\psi_0}\right]-\nonumber \\4\expval{\bm{\eta}_i\cdot\bm{J}}{\psi_0}\expval{\bm{\eta}_j\cdot\bm{J}}{\psi_0},
\end{align}
with $i, j \in \{1,2\}$.
Multiparameter QCRB is a better metrological bound compared to HCRB and moreover, unlike the single-qubit probe case, we can saturate the multiparameter QCRB here using the weak commutation condition~\cite{Candeloro_2024}, since we are not inhibited by the dimension of the probe anymore. Therefore, we will use multiparameter QCRB in this case to find the optimal probe.
We need to minimize $\Tr\left(W\mathcal{F}^{-1}_Q\right)$ under the constraint that the encoded state must satisfy the weak commutativity condition given in Eq.~\eqref{weak}, thereby ensuring that the multiparameter QCRB becomes an achievable bound.
To do this, for the sake of simplicity, we choose $\bm{a}_1=(1,0,0)^T$ and $\bm{a}_2=\left(\cos{(\theta)},\sin{(\theta)},0\right)^T$ with $\theta \in \left[0, \pi\right]$.

\subsubsection{Choosing an ansatz probe}
In order to find the optimal probe in this estimation task, we first consider a special class of probes which satisfy the weak commutation condition written in Eq.~\eqref{weak} (that will show later). Then, in the latter part of the paper, we will prove through re-parametrization of the estimation problem that other probes which satisfy the weak commutation condition cannot provide better precision compared to our chosen probe in this case when we use multiparameter QCRB as our metrological bound. We consider the following choice of a single-qutrit probe:
\begin{equation}
\label{probe-qutrit}
    \ket{\psi_0} = \cos{(\alpha)}\ket{\lambda_{\text{max}}}+e^{i\psi} \sin{(\alpha)}\ket{\lambda_{\text{min}}}
\end{equation}
with $\alpha\in [0,\pi/2]$ and $\psi \in [-\pi,\pi]$, where $\ket{\lambda_{\text{max}}} , \ket{\lambda_{\text{min}}}$ denote the highest and lowest eigenstates of the total encoding Hamiltonian respectively. 
Note that the total encoding Hamiltonian can be written as
\begin{align*}
H(\bm{\varphi})=\sqrt{\varphi^2_1+\varphi^2_2+2\varphi_1\varphi_2\cos{(\theta)}}\bigg[\cos{(\phi)}J_x+\sin{(\phi)}J_y\bigg],
\end{align*}
where 
\begin{align*}
    \cos{(\phi)}=\frac{\varphi_1+\varphi_2\cos{(\theta)}}{\sqrt{\varphi^2_1+\varphi^2_2+2\varphi_1\varphi_2\cos{(\theta)}}},\\
    \sin{(\phi)}=\frac{\varphi_2\sin{(\theta)}}{\sqrt{\varphi^2_1+\varphi^2_2+2\varphi_1\varphi_2\cos{(\theta)}}}.
\end{align*}
The Hamiltonian can then be easily expressed using the Baker-Campbell-Hausdorff formula as
\begin{equation*}
H(\bm{\varphi})=\sqrt{\varphi^2_1+\varphi^2_2+2\varphi_1\varphi_2\cos{(\theta)}}\bigg[e^{-i\phi J_z}J_xe^{i\phi J_z}\bigg].
\end{equation*}
The operator in the square brackets represents a rotated $J_x$ operator. Hence, if $\ket{j;m_x}$ represents an eigenstate of $J_x$ with an eigenvalue $m_x$, then the corresponding eigenstate of the rotated operator would be $e^{-i\phi J_z}\ket{j;m_x}$. Using Wigner's d-matrices, the eigenvectors of $J_x$ can be expressed in the $J_z$ basis as
\begin{align*}
    \ket{j,m_x=j} &= \frac{1}{2^j}\sum_{m_z=-j}^{m_z=j}\binom{2j}{j+m_z}^{1/2}\ket{j,m_z},\\
    \ket{j,m_x=-j} &= \frac{1}{2^j}\sum_{m_z=-j}^{m_z=j}\binom{2j}{j+m_z}^{1/2}\left(-1\right)^{j+m_z}\ket{j,m_z}.
\end{align*}
Hence, the highest and lowest eigenvectors of our total encoding Hamiltonian are given by
\begin{align*}
    \ket{\lambda_{\text{max}}}&=\frac{1}{2^j}\sum_{m_z=-j}^{m_z=j}\binom{2j}{j+m_z}^{1/2}e^{-i\,m_z\phi}\ket{j;m_z},\\
\ket{\lambda_{\text{min}}}&=\frac{1}{2^j}\sum_{m_z=-j}^{m_z=j}\binom{2j}{j+m_z}^{1/2}\left(-1\right)^{j+m_z}e^{-i\,m_z\phi}\ket{j;m_z},
\end{align*}
with $j = 1$.
Our chosen probe can thus be expressed in terms of the highest and lowest eigenvectors of our total encoding Hamiltonian written in Eq.~\eqref{probe-qutrit} as follows:
\begin{align}
\label{probe-expr}
    &\ket{\psi_0}=\frac{1}{2}\sum_{m_z = -1}^{m_z=1}\binom{2}{1+m_z}^{1/2}e^{-im_z\phi}\bigg[\cos{(\alpha)}\nonumber\\&+e^{i\psi}(-1)^{1+m_z}\sin{(\alpha)}\bigg]\ket{1,m_z}.
\end{align} 
The chosen probe can be rewritten as $\displaystyle{\ket{\psi_0}=\sum_{m_z=-1}^{m_z=1}c_{m_z}\ket{1;m_z}}$, with 
\begin{align*}
    c_0 &= \frac{1}{\sqrt{2}}\bigg[\cos{(\alpha)}-e^{i\psi}\sin{(\alpha)}\bigg],\\
    c_{\pm 1}&=\frac{e^{\mp i\phi}}{2}\bigg[\cos{(\alpha)}+e^{i\psi}\sin{(\alpha)}\bigg],
\end{align*}
where $\left(c^*_{-1}c_0+c^*_0c_1\right)=\displaystyle{\frac{e^{-i\phi}}{\sqrt{2}}\cos{(2\alpha)}}$.

To satisfy the weak-commutation criterion by any probe, the following condition must hold true:
\begin{equation*}
\expval{\left(\bm{\eta}_1\times\bm{\eta}_2\right)\cdot\bm{J}}{\psi_0} = 0.
\end{equation*}
The expectation value of any $\mathbb{SU}(2)$ Hamiltonian of the form $\bm{t}\cdot\bm{J}$ with $\bm{t}\coloneqq(\bm{\eta}_1\times\bm{\eta}_2)$ can be expressed as
\begin{align*}
\expval{\bm{t}\cdot\bm{J}}{\psi_0}&=\sqrt{2}\text{Re}\bigg[\left(t_x+i t_y\right)\left(c^*_{-1}c_0+c^*_0c_1\right)\bigg],\\
&=\cos{(2\alpha)}\left(\bm{t}\cdot\hat{n}_{rot}\right),
\end{align*}
where the form of $\hat{n}_{rot}$ is written in Eq.~\eqref{coherent_normal}.
Now, note that the weak commutation condition is always satisfied for any arbitrary values of $\alpha$ and $\psi$. Because we have already proved in the previous subsection that $\left(\bm{\eta}_1\times\bm{\eta}_2\right)\cdot\hat{n}_{rot}=0$, and it will also hold true here since $\bm{\eta}_1$ and $\bm{\eta}_2$ are the same in both cases.
Thus, our chosen class of probes written in Eq.~\eqref{probe-expr}, will satisfy the weak commutation condition.
Hence, our chosen class of probes is a good choice to begin with for finding the optimal probe in this estimation task.

\subsubsection{Re-parametrization}
In order to find the optimal single-qutrit probe which minimizes the multiparameter QCRB in this two-phase $\mathbb{SU}(2)$ encoding estimation task, we can try to re-parameterize our problem so that the expressions of $\bm{\eta}_{1,2}$ are simpler to deal with. Instead of estimating the Hamiltonian strengths $\varphi_1,\varphi_2$, we can alternatively consider estimating the overall magnitude ($B(\varphi_1,\varphi_2)$) of rotation caused by the unitary evolution and the effective angle ($\phi(\varphi_1,\varphi_2)$) that the rotation axis makes with the $x$-direction, i.e., 
\begin{align*}
    B(\varphi_1,\varphi_2) &= \sqrt{\varphi^2_1+\varphi^2_2+2\varphi_1\varphi_2\cos{(\theta)}},\\
    \phi(\varphi_1,\varphi_2)&=\begin{cases}
        \tan^{-1}{\left(\frac{\varphi_2\sin{(\theta)}}{\varphi_1+\varphi_2\cos{(\theta)}}\right)},\varphi_1+\varphi_2\cos{(\theta)} > 0\\
        \pi + \tan^{-1}{\left(\frac{\varphi_2\sin{(\theta)}}{\varphi_1+\varphi_2\cos{(\theta)}}\right)}\\\quad\quad\quad\quad,\varphi_1+\varphi_2\cos{(\theta)} < 0.
    \end{cases}
\end{align*}

In this new parametrization, the corresponding effective Hamiltonians which are used to express the QFIM elements are $\mathcal{H}_{\phi}=\bm{\eta}_{\phi}\cdot\bm{J},\mathcal{H}_B=\bm{\eta}_{B}\cdot\bm{J}$. The expressions of $\bm{\eta}_{\phi},\bm{\eta}_{B}$ are written in Appendix~\ref{qfim_repr}.
Since the weak commutation condition $D=0$ is re-parametrization independent~\cite{Carollo_2019,paris2022}, our chosen probe will also satisfy the weak commutation condition in the re-parameterized estimation problem as well.

\subsubsection{Finding the optimal single-qutrit probe}
We now find the elements of the QFIM in the new parametrization and try to look for conditions on our chosen probe when the precision will be maximized for any chosen weight matrix. Keeping the details of the calculation in the Appendix \ref{qfim_repr}, we only write the elements of the re-parameterized QFIM here as follows:
\begin{align*}
\left(\mathcal{\tilde{F}}_Q\right)_{\phi,\phi}&=8\sin^2{\left(\frac{B}{2}\right)}\bigg[1-\sin{(2\alpha)}\cos{(\psi+B)}\bigg],\\
\left(\mathcal{\tilde{F}}_Q\right)_{B,B}&=4\sin^2{\left(2\alpha\right)}, \quad \left(\mathcal{\tilde{F}}_Q\right)_{\phi,B}=\left(\mathcal{\tilde{F}}_Q\right)_{ B,\phi}=0.
\end{align*}
Note that the QFIM is diagonal in this parametrization for our chosen probe. Furthermore, if we choose $\alpha^{\text{max}}=\frac{\pi}{4}$ and $\psi^{\text{max}}=\pi-B$, the diagonal elements of this re-parameterized QFIM for this probe, denoted by $\mathcal{\tilde{F}}_Q^{\text{max}}$, is given by
\begin{align} 
\label{QFIM-qutrit-1}
\left(\mathcal{\tilde{F}}_Q\right)^{\text{max}}_{\phi,\phi}&=16\sin^2{\left(\frac{B}{2}\right)}=4\left(\Delta^2\mathcal{H}_{\phi}\right)_{\text{max}},\\
\left(\mathcal{\tilde{F}}_Q\right)^{\text{max}}_{B,B}&=4 = 4 \left(\Delta^2 \mathcal{H}_B\right)_{\text{max}} \label{QFIM-qutrit-2}.
\end{align}
These diagonal elements are exactly equal to $4$ times the maximum possible variances of the Hamiltonians $\mathcal{H}_{\phi}$ and $\mathcal{H}_B$ calculated over any arbitrary single-qutrit input state and hence are equal to the maximum QFIs corresponding to $\phi$ and $B$.

Now, let us consider any arbitrary pure single-qutrit state which does not belong to the class of states assumed in Eq.~\eqref{probe-qutrit} and let the re-parameterized QFIM corresponding to this state be denoted as $\tilde{\mathcal{G}}_{Q}$. The QCRB corresponding to $\tilde{\mathcal{G}}_Q$ is given by
\begin{align*}
\Tr{\left(\mathcal{W}\tilde{\mathcal{G}}^{-1}_Q\right)}=\frac{w_{11}\left(\tilde{\mathcal{G}}_Q\right)_{BB}+w_{22}\left(\tilde{\mathcal{G}}_Q\right)_{\phi\phi}-2w_{12}\left(\tilde{\mathcal{G}}_Q\right)_{\phi B}}{\left(\tilde{\mathcal{G}}_Q\right)_{\phi\phi}\left(\tilde{\mathcal{G}}_Q\right)_{BB}-\left(\tilde{\mathcal{G}}_Q\right)^2_{\phi B}},
\end{align*}
for the weight matrix expressed in Eq.~\eqref{weight}.

We now have the following two cases:

\textbf{Case-1:} If the choice of our weight matrix is such that $w_{12}\left(\tilde{\mathcal{G}}_Q\right)_{\phi B}\leq0$, we have the following chain of inequalities:
\begin{align*} &\Tr{\left(\mathcal{W}\tilde{\mathcal{G}}^{-1}_Q\right)}\\&= \frac{w_{11}\left(\tilde{\mathcal{G}}_Q\right)_{BB}+w_{22}\left(\tilde{\mathcal{G}}_Q\right)_{\phi\phi}+\Big|2w_{12}\left(\tilde{\mathcal{G}}_Q\right)_{\phi B}\Big|}{\left(\tilde{\mathcal{G}}_Q\right)_{\phi \phi}\left(\tilde{\mathcal{G}}_Q\right)_{BB}-\left(\tilde{\mathcal{G}}_Q\right)^2_{\phi B}},\\
    &\geq \frac{w_{11}\left(\tilde{\mathcal{G}}_Q\right)_{BB}+w_{22}\left(\tilde{\mathcal{G}}_Q\right)_{\phi\phi}}{\left(\tilde{\mathcal{G}}_Q\right)_{\phi \phi}\left(\tilde{\mathcal{G}}_Q\right)_{BB}},\\
    &\geq \frac{w_{11}}{(\tilde{\mathcal{F}}_Q)^{\text{max}}_{\phi \phi}}+\frac{w_{22}}{(\tilde{\mathcal{F}}_Q)^{\text{max}}_{BB}}
\end{align*}
Thus, when $w_{12}\left(\tilde{\mathcal{G}}_Q\right)_{\phi B}\leq0$, we can clearly see that the single-qutrit probe written in Eq.~\eqref{probe-qutrit} with $\alpha=\frac{\pi}{4}$ and $\psi=\pi-B$ leads to the minimum value of the QCRB.

\textbf{Case-2:}
However, if we have a weight matrix such that $w_{12}\left(\tilde{\mathcal{G}}_Q\right)_{12}>0$, then we are unable to say anything about the optimality of the single-qutrit probe by the above procedure.
For any arbitrary pure single-qutrit probe, $(\tilde{\mathcal{F}}^{\text{max}}_Q - \tilde{\mathcal{G}}_Q)$ has non-negative diagonal elements. Consequently, its trace is non-negative, as $\tilde{\mathcal{F}}^{\text{max}}_Q$ has the maximum possible diagonal values achievable in this estimation scenario.
Any general pure single-qutrit can be defined as
\begin{align*}
\ket{\psi_0}&=\cos(\theta_0)\ket{1;0}+\sin(\theta_0)\cos(\theta_1) e^{i\phi_1}\ket{1;1}\\&+\sin(\theta_0)\sin(\theta_1)e^{i\phi_{-1}}\ket{1;-1}
\end{align*}
with $\theta_0, \theta_1 \in \left[0,\pi\right]$ and $\phi_1, \phi_{-1} \in \left[0,2\pi\right]$.
By varying $\theta_0,\theta_1,\phi_1,\phi_{-1}$ along with the values of $B$ from $0$ to $2\pi$ and $\phi$ from $-\pi$ to $\pi$,
we find numerically that the determinant of $(\displaystyle{\tilde{\mathcal{F}}^{\text{max}}_Q}-\tilde{\mathcal{G}}_Q)$ is non-negative across all values of all these parameters. Thus, we always have $\tilde{\mathcal{F}}^{\text{max}}_Q\succeq\tilde{\mathcal{G}}_Q$ in this estimation scenario and hence $\Tr{\left(\mathcal{W}\left(\tilde{\mathcal{F}}^{\text{max}}_Q\right)^{-1}\right)}\leq\Tr{\left(\mathcal{W}\tilde{\mathcal{G}}^{-1}_Q\right)}$. Thus, even in this case, the probe given by
\begin{align}
\label{best-probe-qutrit}
    \ket{\psi}^{\text{opt}}&=\sin{\left(\frac{B}{2}\right)}\bigg[\frac{e^{i\phi}\ket{1;-1}+e^{-i\phi}\ket{1;1}}{\sqrt{2}}\bigg]\nonumber&&\\&\quad\quad+e^{-i\pi/2}\cos{\left(\frac{B}{2}\right)}\ket{1;0}.
\end{align}
turns out to be the optimal input state for simultaneous estimation of the two $\mathbb{SU}(2)$ parameters $\phi$ and $B$.


Since the QFIM is convex in $\rho_{\bm{\varphi}}$, a mixed probe cannot lead to a lower QCRB value than the best pure probe in its mixture~\cite{Liu_2020}. 
Since the optimal pure single-qutrit probe in this case, i.e, when we use $\alpha=\pi/4$ and $\psi=(\pi-B)$ in Eq.~\eqref{best-probe-qutrit} is independent of the choice of weight matrix, we can say that the single-qutrit probe written in Eq.~\eqref{best-probe-qutrit} will be the optimal probe over all single-qutrit probes for any re-parametrization of the two-phase estimation problem as well.

Whenever $B\to 0$ or $B\to 2\pi$, we have
    $\ket{\psi}^{\text{opt}}\approx \ket{1;0}$,
while when $B\to \pi$,
$\ket{\psi}^{\text{opt}}\approx\frac{e^{i\phi}\ket{1;-1}+e^{-i\phi}\ket{1;1}}{\sqrt{2}}$.

Moreover, from Eqs.~\eqref{QFIM-qutrit-1} and~\eqref{QFIM-qutrit-2}, we can see that the optimal single-qutrit probe written in Eq.~\eqref{best-probe-qutrit} maximizes the determinant of the re-parameterized QFIM while simultaneously maximizing two diagonal entries of QFIM. Note that the determinants of the QFIMs in two different parametrizations only differ via the square of the determinant of the Jacobian, which in turn only depends on the values of the parameters of interest and not the state-parameters $\alpha,\psi$. Therefore, the optimal probe written in Eq.~\eqref{best-probe-qutrit}  will maximize the determinant of QFIM of any two $\mathbb{SU}(2)$ encoded phases in this estimation scenario.

Now, we see how the multiparameter QCRB varies for the optimal single-qutrit probe as the angle between the encoding axes $\theta$ is varied from $0$ to $\pi$ in Fig.~\ref{qcrb_qutrit}. We also use the weight matrix with elements $w_{11}=w_{22}=1.0$ and $w_{12}=0.2$ here. The cost of two-$\mathbb{SU}(2)$ phase estimation rises steeply whenever the encoding Hamiltonians become highly commuting here as well, like the previous single-qubit probe case.

\begin{figure}[h!]
    \centering
    \includegraphics[width=1.0\linewidth]{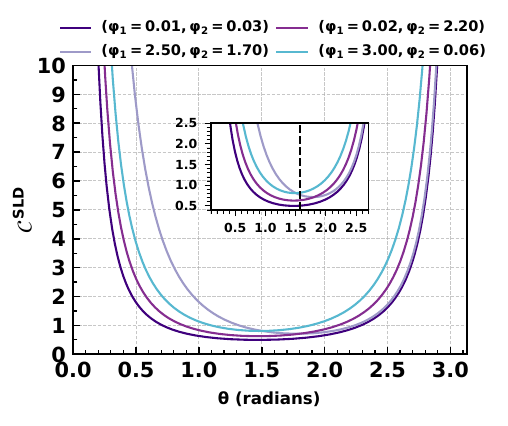}
    \caption{\textbf{Effect of the noncommutativity between $H_1$ and $H_2$ on the optimal HCRB for single-qutrit probes.}
    We study the variation of multiparameter QCRB corresponding to optimal single-qutrit probe in the two $\mathbb{SU}(2)$ phase estimation using single-qutrit probes, with changing $\bm{\theta}$. We plot the optimal/minimal values of the multiparameter QCRB with the angle $\theta$ between the two encoding axes $\bm{a}_1$ and $\bm{a}_2$ which is varied from $0$ to $\pi$. 
    The overall nature of the plot is similar to the one we obtained in the single-qubit case; however, the estimation error here is lower compared to the single-qubit probe estimation scenario. The inset, which shows the zoomed-in picture of the minima corresponding to the different curves, depicts once again that non-orthogonal encodings may outperform orthogonal encodings in certain situations. The vertical axis is dimensionless, while the horizontal axis has a unit of radians.}
    \label{qcrb_qutrit}

\end{figure}

\subsection{Unitary encoding with commuting Hamiltonians and single-qutrit probes}
~\label{commuting-qutrit}
In the subsection, we deal with single-qutrit probes and two-unitary phase estimation but the encoding Hamiltonians, let's say $H_1$ and $H_2$, commute with each other. 
Any normalized pure single-qutrit probe state can be expressed as
\begin{equation*}
\label{qutrit_probe}
    |\psi_0\rangle = \sum_{\mu=0}^2 c_{\mu}|\mu\rangle
\end{equation*}
with $ \sum_{i=0}^2 |c_\mu|^2 = 1$, where $\{|\mu\rangle\}$ is an orthonormal basis and specifically, the basis in which both $H_1$ and $H_2$ are diagonalized.
We consider the eigenspectra of the two Hamiltonians as $\vec{\lambda}\coloneqq \{\lambda_1, \lambda_2, \lambda_3\}$ and $\vec{\lambda}' \coloneqq \{\lambda'_1, \lambda'_2, \lambda'_3\}$, and hence their diagonal representations in this particular basis is given by
\begin{equation}
\label{commuting}
    H_{1}=\begin{pmatrix}
        \lambda_1 & 0 & 0\\
        0 & \lambda_2 & 0\\
        0 & 0 & \lambda_3
    \end{pmatrix}, \quad H_{2}=\begin{pmatrix}
        \lambda_1' & 0 & 0\\
        0 & \lambda_2' & 0\\
        0 & 0 & \lambda_3'
    \end{pmatrix},
\end{equation}
where $\vec{\lambda}, \vec{\lambda'} \in \mathbb{R}^3$.
If our Hamiltonians commute with each other, the weak-commutation condition written in Eq.~\eqref{weak} is automatically satisfied and requires no further constraints to be imposed on probes to find the optimal one. Thus, in this case, we always have $\mathcal{\mathbf{C}}^{\textbf{Holevo}}=\mathcal{\mathbf{C}}^{\textbf{SLD}}$ and we can saturate the QCRB.
We show here that the maximal quantum coherence in the eigenbasis of the total encoding Hamiltonians is not necessary in order to achieve optimal precision through an example. In this example, the encoding Hamiltonians are chosen such that their eigenvalue spectra satisfy the following relations:
\begin{equation*}
 \hspace{3 mm} \lambda_1 = \lambda_3 = \lambda_2 - 1, \quad \lambda’_2 = \lambda’_3 = \lambda’_1 - 1.
\end{equation*}
For this specific case, with a weight matrix of the form in Eq.\eqref{weight}, one can easily find that the multiparameter quantum Cram\'er-Rao bound is expressed as
\begin{equation*}
    \mathcal{\mathbf{C}}^{\textbf{SLD}}=\frac{1}{4}\left[\frac{w_{22}}{|c_0|^2}+\frac{w_{11}}{|c_1|^2}+\frac{w_{11}+w_{22}+2w_{12}}{(1-|c_0|^2-|c_1|^2)}\right]
\end{equation*}
We again choose a weight matrix with elements $w_{11}=w_{22}=1$ and $w_{12}=0.2$ like we did in the previous $\mathbb{SU}(2)$ encoding cases.
From here, we can use the second partial derivative test to find that the minima of the multiparameter QCRB is achieved when
\begin{align*}
    |c_0|&=\sqrt{\frac{\sqrt{5}}{4}\left(\sqrt{5}-\sqrt{3}\right)}\approx 0.5308\\
    |c_1|&=\sqrt{\frac{5}{4}\left(\sqrt{5}-\sqrt{3}\right)}\approx 0.7937\\
    |c_2|&=\sqrt{1-|c_0|^2-|c_1|^2}\approx 0.2970
\end{align*}
which is indeed a completely coherent state, but not maximally coherent in the eigenbasis of encoding Hamiltonians.

The reason that the optimal state is completely coherent is that if we had any one of the coefficients $c_0, c_1$ or $c_2$ equal to zero, we would essentially be working in a two-dimensional subspace, as the commuting Hamiltonians would not intermix their eigenvectors. Since we know that two commuting unitary encodings with a single-qubit probe always lead to a singular QFIM, this would hold true here as well for the single-qutrit case if we had an incompletely coherent probe. 
However, this is in contrast to what we obtained in the single-qubit and single-qutrit cases using $\mathbb{SU}(2)$ encodings. Neither maximal quantum coherence nor the equal superposition of the highest and lowest eigenstates acts as an optimal probe in this case. It suggests that, unlike single-parameter unitary encoding estimation, quantum coherence in the eigenbasis of encoding Hamiltonians is not the real resource in deciding the optimal precision. 



\section{conclusion}
\label{conclude}
Determination of the optimal probe for any quantum metrological task is of utmost importance since 
it provides us the information about the best precision achievable in the estimation task.
In particular, for single-phase estimation problems, it is widely known that the optimal probe is a pure state comprising an equally weighted superposition of the eigenvectors of the encoding Hamiltonian corresponding to its highest and lowest eigenvalues and the relative phase in this superposition can be arbitrary. 
The choice of this probe is motivated by the fact that it leads to the maximum quantum Fisher information associated with the encoded phase. This fact can then be further utilized to prove that if instead entangled states of multiple probes are employed in this single unitary phase estimation task, the optimal precision is achieved by a {GHZ-like} state. Such a state allows the estimation precision to surpass the classical shot-noise scaling and reach the so-called Heisenberg limit, resulting in a quadratic improvement in estimation accuracy. 


However, when multiple parameters of a quantum statistical model are to be simultaneously estimated, the structure of optimal probes remains less explored. In this work, we have identified the optimal probe states for the simultaneous estimation of two arbitrary phases encoded via $\mathbb{SU}(2)$ unitaries for arbitrary weight matrices, using two classes of quantum probes—(i) single-qubit and (ii) single-qutrit.
Previous works in this direction have already pointed out that, for single-qubit probes, estimation of two or more unitarily encoded parameters can never saturate the multiparameter quantum Cram\'er-Rao bound~\cite{paris2022,Carollo_2019,Candeloro_2024}. Hence, in this setting, we have employed the Holevo Cram\'er-Rao bound as the fundamental lower bound for the error in the estimation since it is an achievable bound for pure probes. In contrast, for single-qutrit probes, the QCRB is known to be saturable under suitable conditions, making it the appropriate bound for optimality. For both single-qubit and single-qutrit cases, we have found that the optimal probe is a unique pure state that maintains the structure of an equally weighted superposition of the eigenvectors of the total encoding Hamiltonian corresponding to the maximum and minimum eigenvalues, mirroring the single-phase estimation case.
However, unlike in the single-phase scenario, the relative phase in the superposition is fixed by the estimation task in both the single-qubit and single-qutrit cases.
Moreover, the optimal probe in both cases not only minimizes the HCRB or QCRB, respectively, but also maximizes the determinant of the quantum Fisher information matrix and is also independent of the choice of weight matrix.
Therefore, in both cases, the optimal state will show the best precision for the estimation of any two arbitrary $\mathbb{SU}(2)$ encoded parameters that are re-parametrizations of our two phases.
We also demonstrate that the optimal precision in estimating two $\mathbb{SU}(2)$-encoded phases using both single-qubit and single-qutrit probes is strongly influenced by the degree of commutativity between the two encoding generators: while increased commutativity leads to a significant loss in precision, maximal noncommutativity between the generators does not yield the optimal precision.
Furthermore, when two unitary parameters are encoded using commuting Hamiltonians on single-qutrit probes, the optimal probe is found to be a state that exhibits complete coherence in the eigenbasis of the total encoding Hamiltonians, although maximal coherence in that basis is not necessary to achieve optimal precision.
\section*{Acknowledgements}
P. Ghosh acknowledges support from the ‘INFOSYS scholarship for senior students’ at Harish-Chandra Research Institute, India. A. Ghoshal acknowledges the support from the Alexander von Humboldt Foundation. 

\twocolumngrid
\bibliography{references}
\onecolumngrid
\section*{Appendix}
\appendix

\section{\texorpdfstring{$\mathbb{SU}(2)$}{SU(2)} encoding with single-qubit probes}
In this section of the appendix, we include some of the detailed derivations related to the simultaneous estimation of the two $\mathbb{SU}
(2)$ encoded phases using a single qubit probe state. In the first part of Appendix A, we derive the expressions for the effective Hamiltonians using Wilcox's formula given in Eq.\eqref{wilcox}. These effective Hamiltonians enable us to express the elements of the QFIM and the Uhlmann matrix in a simple manner and later also turn out to be useful when optimizing to find the best probe state. Thereafter, in the second part of this appendix, we provide the expression for the Uhlmann matrix of an arbitrary qubit mixed state, undergoing a unitary parametrization process. We show that such an Uhlmann matrix may be expressed in terms of the corresponding Uhlmann matrix of either of the two eigen-states of the mixed qubit density matrix with additional terms due to the purity of the state. This helps us prove that a a mixed qubit state can never lead to a lower HCRB value compared to the best pure-state in its mixture.\subsection{Derivation of the expressions for the effective Hamiltonians}
\label{eta_expr}
Using Eq.~\eqref{wilcox}, we can evaluate
\begin{equation}
\mathcal{H}_i=\sum_{n=0}^{\infty}\frac{i^{n+2}}{(n+1)!}\mathcal{\zeta}^{(n)}\left[\frac{1}{2}\bm{a}_i\cdot\bm{\sigma}\right];\quad i=\{1,2\}.
\end{equation}
where $\mathcal{\zeta}^{(n)}[X]$ denotes the stacked commutator as
\begin{equation*}
  \mathcal{\zeta}^{(n)}[X]\equiv \underbrace{\bigg[\frac{1}{2}\varphi_1\bm{a}_1\cdot{\bm{\sigma}}\\+\frac{1}{2}\varphi_2\bm{a}_2\cdot{\bm{\sigma}},\cdots,\bigg[\frac{1}{2}\varphi_1\bm{a}_1\cdot{\bm{\sigma}}\\+\frac{1}{2}\varphi_2\bm{a}_2\cdot{\bm{\sigma}},X\bigg]\cdots\bigg]}_{n\text{ commutators}}
\end{equation*}

$\mathcal{\zeta}^{(n)}\left[\frac{1}{2}\bm{a}_i\cdot\bm{\sigma}\right]$ for each $n$ can be simplified as
\begin{equation}
\zeta^{(n)}\bigg[\frac{1}{2}\bm{a}_1 \cdot \bm{\sigma}\bigg] =
\begin{cases}
        \frac{1}{2}\bm{a}_1 \cdot \bm{\sigma}, & \text{if } n = 0, \\[8pt]
        \,\varphi_2 \left(\varphi_1 f_2 + \varphi_2 f_1\right)^{\frac{n}{2}-1} (f_1 \bm{a}_1 - f_2 \bm{a}_2)\cdot\frac{\bm{\sigma}}{2}, & \text{if } n \text{ is even}, \\[8pt]
        -i\varphi_2 (\varphi_1 f_2 + \varphi_2 f_1)^{\frac{(n-1)}{2}} (\bm{a}_1 \times \bm{a}_2)\cdot\frac{\bm{\sigma}}{2}, & \text{if } n \text{ is odd}.
    \end{cases}
\end{equation}
\vspace{1pt}
Using the above expressions, after some simple algebra, we can find that the sum of all the even $n$ terms apart from $n=0$ 
can be given as
\begin{align*}
&\frac{\varphi_2}{\left(\varphi_1 f_2+\varphi_2 f_1\right)^{3/2}}\sum_{m=1}^{\infty}(-1)^{m+1} \frac{\left(\sqrt{\varphi_1f_2+\varphi_2 f_1}\right)^{2m+1}}{(2m+1)!}\left(f_1 \bm{a}_1-f_2\bm{a}_2\right)\cdot\frac{\bm{\sigma}}{2}, \\&= \frac{\varphi_2}{\left(\varphi_1 f_2+\varphi_2 f_1\right)^{3/2}}\left[\sqrt{\varphi_1f_2+\varphi_2f_1}-\sin{\left(\sqrt{\varphi_1f_2+\varphi_2 f_1}\right)}\right]\left(f_1\bm{a}_1-f_2\bm{a}_2\right)\cdot\frac{\bm{\sigma}}{2},
\end{align*}
where $f_1(\bm{\varphi}) \coloneqq \varphi_1\cos{(\theta)}+\varphi_2$, $f_2(\bm{\varphi}) \coloneqq \varphi_1+\varphi_2\cos{(\theta)}$.
The sum of all the odd $n$ terms is given by
\begin{align*}
    &\frac{\varphi_2}{\left(\varphi_1f_2+\varphi_2f_1\right)}\sum_{m=0}^{\infty} (-1)^{m+1}\frac{\left(\sqrt{\varphi_1f_2+\varphi_2f_1}\right)^{2m+2}}{(2m+2)!}\,\left(\bm{a}_1\times\bm{a}_2\right)\cdot\frac{\bm{\sigma}}{2} =-\frac{2\varphi_2}{\varphi_1f_2+\varphi_2f_1}\sin^2{\left(\frac{\sqrt{\varphi_1f_2+\varphi_2 f_1}}{2}\right)}\left(\bm{a}_1\times\bm{a}_2\right)\cdot\frac{\bm{\sigma}}{2}
\end{align*}
Using this, we can write $\mathcal{H}_1$ as
\begin{equation}
 \mathcal{H}_1=\frac{1}{2}\bm{\eta}_1\cdot\bm{\sigma},
\end{equation}
where
\label{eta expression}
\begin{multline}\label{eta1}
    \bm{\eta}_1=-\bm{a}_1+\frac{\varphi_2}{\left(\varphi_1 f_2+\varphi_2 f_1\right)^{3/2}}\left[\sqrt{\varphi_1f_2+\varphi_2f_1}-\sin{\left(\sqrt{\varphi_1f_2+\varphi_2 f_1}\right)}\right]\left(f_1\bm{a}_1-f_2\bm{a}_2\right)\\-\frac{2\varphi_2}{\varphi_1f_2+\varphi_2f_1}\sin^2{\left(\frac{\sqrt{\varphi_1f_2+\varphi_2 f_1}}{2}\right)}\left(\bm{a}_1\times\bm{a}_2\right).
\end{multline}
Similarly, $\mathcal{H}_2$  can be simplified, which is $\mathcal{H}_2=\frac{1}{2}\bm{\eta}_2\cdot\bm{\sigma},$ where
\begin{multline}
\bm{\eta}_2=-\bm{a}_2-\frac{\varphi_1}{\left(\varphi_1 f_2+\varphi_2 f_1\right)^{3/2}}\left[\sqrt{\varphi_1f_2+\varphi_2f_1}-\sin{\left(\sqrt{\varphi_1f_2+\varphi_2 f_1}\right)}\right]\left(f_1\bm{a}_1-f_2\bm{a}_2\right)\\+\frac{2\varphi_1}{\varphi_1f_2+\varphi_2f_1}\sin^2{\left(\frac{\sqrt{\varphi_1f_2+\varphi_2 f_1}}{2}\right)}\left(\bm{a}_1\times\bm{a}_2\right). 
\end{multline}

Using the expressions for $\bm{\eta}_{1,2}$, the inverse of the QFIM can be easily found, and hence the expression for $\mathcal{Q}$ is simple to obtain. 

\subsection{Uhlmann matrix for mixed state with unitary parametrization}
\label{purity-uhlmann}
Let us consider an arbitrary mixed $d$-dimensional quantum state which is used as a probe. Its spectral decomposition is given by
\begin{align}
    \rho=\sum_{\mu=0}^{s-1}\lambda_{\mu}\ketbra{\psi_{\mu}}{\psi_{\mu}},
\end{align}
where $s<d$ denotes the dimension of the support space of $\rho$.
A unitary parametrization operation $U(\bm{\varphi})$ encodes the parameters $\bm{\varphi}$ on to this state, such that the encoded state is given by
\begin{align}
    \rho_{\bm{\varphi}}=U(\bm{\varphi})\rho U^{\dagger}(\bm{\varphi})=\sum_{\mu=0}^{s-1}\lambda_{\mu}\ketbra{\psi_{\mu}(\bm{\varphi})}{\psi_{\mu}(\bm{\varphi})}.
\end{align}
Corresponding to the $i$-th encoded parameter $\varphi_i$, the matrix elements of the SLD operator $L_{\varphi_i}$ in the eigen-basis of the encoded state is well known \cite{Liu_2014} as
\begin{align*}
\left(L_{\varphi_i}\right)_{\mu\nu}=\begin{cases}
\frac{2\left(\partial_{\varphi_i}\rho(\bm{\varphi})\right)_{\mu\nu}}{\lambda_{\mu}(\bm{\varphi})+\lambda_{\nu}(\bm{\varphi})}; \quad \mu,\nu\in\text{Support}(\rho_{\bm{\varphi}})&\\
\text{indeterminate};\quad\text{otherwise},
\end{cases}   
\end{align*}
where the indeterminate elements corresponding to the case when both $\mu$ and $\nu$ lie outside the support space of $\rho(\bm{\varphi})$ do not contribute to either the QFIM or the Uhlmann matrix. Since the eigenvalues of the encoded state do not depend on the encoded parameters for a unitary encoding, these matrix elements can be written as
\begin{align*}
    \left(L_{\varphi_i}\right)_{\mu\nu}=2\left(\frac{\lambda_{\mu}-\lambda_{\nu}}{\lambda_{\mu}+\lambda_{\nu}}\right)\left(\frac{\partial\bra{\psi_{\mu}(\bm{\varphi})}}{\partial \varphi_i}\right)\ket{\psi_{\nu}(\bm{\varphi})}
\end{align*}
We can now easily see that a term like $\Tr(\rho_{\bm{\varphi}}L_{\varphi_i}L_{\varphi_j})$ can be expressed as follows:
\begin{align}
\Tr(\rho_{\bm{\varphi}}L_{\varphi_i}L_{\varphi_j})=\sum_{\mu=0}^{s-1}\sum_{\nu=0}^{d-1}\lambda_{\mu}\left(L_{\varphi_i}\right)_{\mu\nu}\left(L_{\varphi_j}\right)_{\nu\mu}.
\end{align}
Using this, the $ij$-th element of the Uhlmann matrix can be written as
\begin{align}
    D_{ij}=\frac{1}{2i}\Tr(\rho_{\bm{\varphi}}\left[L_{\varphi_i},L_{\varphi_j}\right])=4\sum_{\mu=0}^{s-1}\sum_{\nu=0}^{d-1}\lambda_{\mu}\left(\frac{\lambda_{\mu}-\lambda_{\nu}}{\lambda_{\mu}+\lambda_{\nu}}\right)^2\text{Im}\left[\left(\frac{\partial\bra{\psi_{\mu}(\bm{\varphi})}}{\partial \varphi_i}\right)\ket{\psi_{\nu}(\bm{\varphi})}\bra{\psi_{\nu}(\bm{\varphi})}\left(\frac{\partial\ket{\psi_{\mu}(\bm{\varphi})}}{\partial \varphi_j}\right)\right].
\end{align}
Similar to the steps followed in Ref.~\cite{Liu_2014}, we can break the summation over $\nu$ into two sums, one from $0$ to $s-1$ and the other from $s$ to $d-1$ and then using the completeness of the projectors in the spectral decomposition of $\rho_{\bm{\varphi}}$ one can arrive the following forms of elements of the Uhlmann matrix:
\begin{align*}
    D_{ij}=4\sum_{\mu=0}^{s-1}\lambda_{\mu}\text{Im}\left[\left(\frac{\partial\bra{\psi_{\mu}(\bm{\varphi})}}{\partial \varphi_i}\right)\left(\frac{\partial\ket{\psi_{\mu}(\bm{\varphi})}}{\partial \varphi_j}\right)\right]-16\sum_{\mu=0}^{s-1}\sum_{\nu=0}^{s-1}\frac{\lambda^2_{\mu}\lambda_{\nu}}{\left(\lambda_{\mu}+\lambda_{\nu}\right)^2}\text{Im}\left[\left(\frac{\partial\bra{\psi_{\mu}(\bm{\varphi})}}{\partial \varphi_i}\right)\ket{\psi_{\nu}(\bm{\varphi})}\bra{\psi_{\nu}(\bm{\varphi})}\left(\frac{\partial\ket{\psi_{\mu}(\bm{\varphi})}}{\partial \varphi_j}\right)\right].
\end{align*}
which for unitary parametrization processes can be written using the effective Hamiltonians as
\begin{align}
    D_{ij}=4\sum_{\mu=0}^{s-1}\lambda_{\mu}\text{Im}\left[\matrixel{\psi_{\mu}}{\mathcal{H}_i\mathcal{H}_j}{\psi_{\mu}}\right]-16\sum_{\mu=0}^{s-1}\sum_{\nu=0}^{s-1}\frac{\lambda^2_{\mu}\lambda_{\nu}}{\left(\lambda_{\mu}+\lambda_{\nu}\right)^2}\text{Im}\left[\matrixel{\psi_{\mu}}{\mathcal{H}_i}{\psi_{\nu}}\matrixel{\psi_{\nu}}{\mathcal{H}_j}{\psi_{\mu}}\right].
\end{align}
Now, specifically for a qubit mixed state, we have $s=2$, hence we have 
\begin{align}
D_{ij}&=4\lambda_0\text{Im}\left[\matrixel{\psi_0}{\mathcal{H}_i\mathcal{H}_j}{\psi_{0}}\right]+4\lambda_1\text{Im}\left[\matrixel{\psi_1}{\mathcal{H}_i\mathcal{H}_j}{\psi_{1}}\right]-16\lambda^2_0\lambda_1\text{Im}\left[\matrixel{\psi_{0}}{\mathcal{H}_i}{\psi_{1}}\matrixel{\psi_{1}}{\mathcal{H}_j}{\psi_{0}}\right]\nonumber,\\&-16\lambda^2_1\lambda_0\text{Im}\left[\matrixel{\psi_{1}}{\mathcal{H}_i}{\psi_{0}}\matrixel{\psi_{0}}{\mathcal{H}_j}{\psi_{1}}\right],
\end{align}
where the $\mu=\nu$ terms from the last summation do not appear here because they do not have imaginary parts. Now using $\mathbb{I}_2=\ketbra{\psi_0}{\psi_0}+\ketbra{\psi_1}{\psi_1}$ the above expression can be greatly simplified since
\begin{align*}
&\text{Im}\left[\matrixel{\psi_0}{\mathcal{H}_i\mathcal{H}_j}{\psi_0}\right]=\text{Im}\left[\matrixel{\psi_0}{\mathcal{H}_i\ket{\psi_1}\bra{\psi_1}\mathcal{H}_j}{\psi_0}\right],\\
&\text{Im}\left[\matrixel{\psi_1}{\mathcal{H}_i\mathcal{H}_j}{\psi_1}\right]=\text{Im}\left[\matrixel{\psi_1}{\mathcal{H}_i\ketbra{\psi_0}{\psi_0}\mathcal{H}_j}{\psi_1}\right]=-\text{Im}\left[\matrixel{\psi_0}{\mathcal{H}_i\mathcal{H}_j}{\psi_0}\right].
\end{align*}
\begin{align}
\implies D_{ij}&=4\left(\lambda_0-\lambda_1\right)\text{Im}\left[\matrixel{\psi_0}{\mathcal{H}_i\mathcal{H}_j}{\psi_{0}}\right]-16\lambda_0\lambda_1\left(\lambda_0-\lambda_1\right)\text{Im}\left[\matrixel{\psi_0}{\mathcal{H}_i \mathcal{H}_j}{\psi_{0}}\right]\nonumber,\\
&=4\left(\lambda_0-\lambda_1\right)\text{Im}\left[\matrixel{\psi_0}{\mathcal{H}_i \mathcal{H}_j}{\psi_{0}}\right]\left(1-4\,\det(\rho)\right).
\end{align}
For a single-qubit quantum state, we know that $\det(\rho)=\left(1-\Tr(\rho^2)\right)/2$. Also, $\left(\lambda_0-\lambda_1\right)^2=\Tr(\rho^2)-2\,\det(\rho)$. Therefore, we can finally write
\begin{equation}
    D_{ij}=\pm 4 \left(2\Tr(\rho^2)-1\right)^{3/2} \text{Im}\left[\matrixel{\psi_0}{\mathcal{H}_i \mathcal{H}_j}{\psi_{0}}\right]=\pm \left(2\Tr(\rho^2)-1\right)^{3/2} D_{ij}(\ketbra{\psi_0(\bm{\varphi})}{\psi_0(\bm{\varphi})}),
\end{equation}
where the $+$ sign is used when $\lambda_0>\lambda_1$ and the $-$ sign is used when $\lambda_1>\lambda_0$.

\section{\texorpdfstring{$\mathbb{SU}(2)$}{SU(2)} encoding with single-qutrit probes}
In this Appendix, we provide the detailed derivations related to the simultaneous estimation of the two $\mathbb{SU}(2)$ phases using a single-qutrit probe state. These results shall complement the discussions in the main text. In the first part of this appendix, we discuss the weak-commutativity of a specific type of qutrit probe. This enables us to validate the effectiveness of such a qutrit probe for our two phase estimation problem. We also provide a useful expression for the expectation value of any $\mathbb{SU}(2)$ operator calculated using this probe. In the second part of the appendix we provide the derivation of the quantum Fisher information matrix elements in the re-parameterized estimation problem and prove its diagonal structure.
\subsection{Weak-Commutation for single-qutrit probes}
\label{qutrit_weak}
Here we prove that the type of qutrit probe we consider prior to optimization, always satisfies the weak-commutation criterion. To start with, consider any operator of the form $H=\bm{t}\cdot\bm{J}$. First we consider the action of this operator on any arbitrary $J_z$ eigenstate as
\begin{align*}   \bm{t}\cdot\bm{J}\ket{1;\mu}&=\bigg[\frac{t_x}{2}\left(J_{+}+J_{-}\right)-i\frac{t_y}{2}\left(J_{+}-J_{-}\right)+t_zJ_z\bigg]\ket{1;\mu}\\
    &=\frac{t_x-it_y}{2}\sqrt{\left(1-\mu\right)\left(2+\mu\right)}\ket{1,\mu+1}+\frac{t_x+it_y}{2}\sqrt{\left(1+\mu\right)\left(2-\mu\right)}\ket{1;\mu-1}+\mu t_z\ket{1;\mu}
\end{align*}
If our probe state is expressed as $\ket{\psi}_0=\sum_{\mu=-1}^{1}c_{\mu}\ket{1;\mu}$, we get
\begin{align*}
    &\expval{\bm{t}\cdot\bm{J}}{\psi_0}=\sum_{\mu}\sum_{\nu}c^*_{\nu}c_{\mu}\matrixel{1;\nu}{\bm{t}\cdot\bm{j}}{1;\mu}\\
    &=\frac{t_x-it_y}{2}\sum_{\mu}\sum_{\nu}c^*_{\nu}c_{\mu}\sqrt{\left(1-\mu\right)\left(2+\mu\right)}\delta_{\nu,\mu+1}+\frac{t_x+it_y}{2}\sum_{\mu}\sum_{\nu}c^*_{\nu}c_{\mu}\sqrt{\left(1+\mu\right)\left(2-\mu\right)}\delta_{\nu,\mu-1}+t_z\sum_{\mu}\sum_{\nu}\mu c^*_{\nu}c_{\mu}\delta_{\mu,\nu}\\
    &=\frac{t_x-it_y}{2}\sum_{\mu}c^*_{\mu+1}c_{\mu}\sqrt{\left(    1-\mu\right)\left(2+\mu\right)}+\frac{t_x+it_y}{2}\sum_{\mu}c^*_{\mu}c_{\mu+1}\sqrt{\left(1-\mu\right)\left(2+\mu\right)}+t_z\sum_{\mu}\mu|c_{\mu}|^2\\
    &=t_z\sum_{\mu}\mu |c_{\mu}|^2+\text{Re}\bigg[\left(t_x+it_y\right)\sum_{\mu}c^*_{\mu}c_{\mu+1}\sqrt{\left(1-\mu\right)\left(2+\mu\right)}\bigg]
\end{align*}
As we can see from Eq.\eqref{probe-expr}, we have $|c_{1}|=|c_{-1}|$, which implies that the first term in the equation above must vanish for this probe and hence
\begin{align*}
    &\expval{\bm{t}\cdot\bm{J}}{\psi_0}=\text{Re}\bigg[\left(t_x+it_y\right)\sum_{\mu}c^*_{\mu}c_{\mu+1}\sqrt{\left(1-\mu\right)\left(2+\mu\right)}\bigg]
=\sqrt{2}\text{Re}\bigg[\left(t_x+it_y\right)\lbrace c^*_{-1}c_0+c^*_{0}c_1\rbrace\bigg]
\end{align*}
We have 
\begin{align*}
    &c^*_{-1}c_0=\frac{e^{-i\phi}}{2\sqrt{2}}\bigg[\cos{(\alpha)}+e^{-i\psi}\sin{(\alpha)}\bigg]\bigg[\cos{(\alpha)}-e^{i\psi}\sin{(\alpha)}\bigg],c^*_{0}c_1=\frac{e^{-i\phi}}{2\sqrt{2}}\bigg[\cos{(\alpha)}-e^{-i\psi}\sin{(\alpha)}\bigg]\bigg[\cos{(\alpha)}+e^{i\psi}\sin{(\alpha)}\bigg]\\
    &\implies c^*_{-1}c_0+c^*_{0}c_1=\frac{e^{-i\phi}}{2\sqrt{2}}\bigg[\cos^2{(\alpha)}-i\sin{(2\alpha)}\sin{(\psi)}-\sin^2{(\alpha)}+\cos^2{(\alpha)}+i\sin{(2\alpha)}\sin{(\psi)}-\sin^2{(\alpha)}\bigg]\\
    &\implies c^*_{-1}c_0+c^*_{0}c_1=\frac{e^{-i\phi}}{\sqrt{2}}\cos{(2\alpha)}
    \end{align*}
    \begin{equation}
    \therefore \expval{\bm{t}\cdot\bm{J}}{\psi_0}=\cos{(2\alpha)}\,\text{Re}\bigg[e^{-i\phi}\left(t_x+it_y\right)\bigg]=\cos{(2\alpha)}\bigg[t_x\cos{\left(\phi\right)}+t_y\sin{\left(\phi\right)}\bigg]=\cos{(2\alpha)}\left(\bm{t}\cdot\hat{n}_{rot}\right)    
    \end{equation}
    Specifically, when we have $\bm{t}=\bm{\eta}_1\times\bm{\eta}_2$, we have
    \begin{equation*}
        \expval{\left(\bm{\eta}_1\times\bm{\eta}_2\right)\cdot\bm{J}}{\psi_0}=\cos{(2\alpha)}\bigg[\left(\bm{\eta}_1\times\bm{\eta}_2\right)\cdot\hat{n}_{rot}\bigg]=0
    \end{equation*}
    where the last equality holds because of the fact that the $\bm{\eta}_{1,2}$ vectors do not change on changing the dimension of our system and in case of qubits we already showed that their cross-product is indeed orthogonal to $\hat{n}_{rot}$.
    \subsection{Quantum Fisher information matrix elements for re-parameterized problem}\label{qfim_repr}
    In this section of the appendix, we derive the expressions for the QFIM elements using a probe of the form given in Eq.\eqref{probe-qutrit} for the re-parameterized problem. As before, the QFIM elements for such a unitary-encoding process is given as
    \begin{equation*}
        \left(\tilde{F}_Q\right)_{ij}=4\bigg[\text{Re}\left(\expval{\tilde{\mathcal{H}}_i\tilde{\mathcal{H}}_j}{\psi_0}\right)-\expval{\tilde{\mathcal{H}}_i}{\psi_0}\expval{\tilde{\mathcal{H}}_j}{\psi_0}\bigg]
    \end{equation*}
    where once again $\tilde{\mathcal{H}}_i=\tilde{\bm{\eta}}_i\cdot\bm{J}$.
    \\
    First we take a look at how the product of two effective Hamiltonians looks like in the case of spin-1 representation.
    \begin{align*}
        &\left(\tilde{\bm{\eta}}_i\cdot\bm{J}\right)\left(\tilde{\bm{\eta}}_j\cdot\bm{J}\right)=\left(\tilde{\eta}_i\right)_x\left(\tilde{\eta}_j\right)_x J^2_x+\left(\tilde{\eta}_i\right)_y\left(\tilde{\eta}_j\right)_y J^2_y+\left(\tilde{\eta}_i\right)_z\left(\tilde{\eta}_j\right)_z J^2_z+\left(\tilde{\eta}_i\right)_x\left(\tilde{\eta}_j\right)_y J_xJ_y+\left(\tilde{\eta}_i\right)_y\left(\tilde{\eta}_j\right)_x J_yJ_x\\
        &+\left(\tilde{\eta}_i\right)_y\left(\tilde{\eta}_j\right)_z J_yJ_z+\left(\tilde{\eta}_i\right)_z\left(\tilde{\eta}_j\right)_y J_zJ_y+\left(\tilde{\eta}_i\right)_x\left(\tilde{\eta}_j\right)_z J_xJ_z+\left(\tilde{\eta}_i\right)_z\left(\tilde{\eta}_j\right)_x J_zJ_x\\&=\bigg[\frac{\left(\tilde{\eta}_i\right)_x\left(\tilde{\eta}_j\right)_x}{4}\left(J_++J_-\right)^2-\frac{\left(\tilde{\eta}_i\right)_y\left(\tilde{\eta}_j\right)_y}{4}\left(J_+-J_-\right)^2+\left(\tilde{\eta}_i\right)_z\left(\tilde{\eta}_j\right)_zJ^2_z\bigg]\\&-\frac{i}{4}\bigg[\left(\tilde{\eta}_i\right)_x\left(\tilde{\eta}_j\right)_y\left(J_++J_-\right)\left(J_+-J_-\right)+
        \left(\tilde{\eta}_i\right)_y\left(\tilde{\eta}_j\right)_x\left(J_+-J_-\right)\left(J_++J_-\right)
        \bigg]\\&-\frac{i}{2}\bigg[\left(\tilde{\eta}_i\right)_y\left(\tilde{\eta}_j\right)_z\left(J_+-J_-\right)J_z+\left(\tilde{\eta}_i\right)_z\left(\tilde{\eta}_j\right)_y J_z\left(J_+-J_-\right)\bigg]+\frac{1}{2}\bigg[\left(\tilde{\eta}_i\right)_x\left(\tilde{\eta}_j\right)_z\left(J_++J_-\right)J_z+\left(\tilde{\eta}_i\right)_z\left(\tilde{\eta}_j\right)_xJ_z\left(J_++J_-\right)\bigg]
    \end{align*}
    We need to take expectation value of each of the terms above with respect to $\ket{\psi_0}$ and then find their real parts. 
    \begin{flalign*}
        &\text{\textbullet}\expval{\left(J_++J_-\right)^2}{\psi_0}=\sum_{\nu}\sum_{\mu}c^*_{\nu}c_{\mu}\matrixel{1;\nu}{\left(J_+\right)^2+\left(J_-\right)^2+\lbrace{J_+,J_-\rbrace}}{1;\mu}&&\\
        &=\sum_{\mu}\sum_{\nu}c^*_{\nu}c_{\mu}\matrixel{1;\nu}{\left(J_+\right)^2+\left(J_-\right)^2+2\left(2-\mu^2\right)}{1;\mu}\\
        &=\sum_{\mu}c^*_{\mu+2}c_{\mu}\sqrt{\left(1-\mu\right)\left(2+\mu\right)}\sqrt{\left(-\mu\right)\left(3+\mu\right)}+\sum_{\mu}c^*_{\mu-2}c_{\mu}\sqrt{\left(1+\mu\right)\left(2-\mu\right)}\sqrt{\left(\mu\right)\left(3-\mu\right)}-2\sum_{\mu}|c_{\mu}|^2\mu^2+4\\
        &=2\bigg[c^*_{1}c_{-1}+c^*_{-1}c_1\bigg]-2\bigg(|c_1|^2+|c_{-1}|^2\bigg)+4\\
        &=4-2|c_{1}-c_{-1}|^2=4-2\sin^2{(\phi)}\bigg[1+\sin{(2\alpha)}\cos{(\psi)}\bigg].
    \end{flalign*}
    Proceeding similarly, one finds
    \begin{flalign*}
        &\text{\textbullet}\matrixel{\psi_0}{\left(J_+-J_-\right)^2}{\psi_0}=2|c_1+c_{-1}|^2-4=2\cos^2{(\phi)}\bigg[1+\sin{(2\alpha)\cos{(\psi)}}\bigg]-4.&&
    \end{flalign*}
    \begin{flalign*}
        &\text{\textbullet}\matrixel{\psi_0}{J^2_z}{\psi_0}=|c_1|^2+|c_{-1}|^2=1-|c_0|^2=1-\frac{1}{2}\bigg[1-\sin{(2\alpha)}\cos{(\psi)}\bigg].&&
    \end{flalign*}
    \begin{flalign*}
        &\text{\textbullet}\matrixel{\psi_0}{\left(J_++J_-\right)\left(J_+-J_-\right)}{\psi_0}=\matrixel{\psi_0}{\left(J_+\right)^2-\left(J_-\right)^2-[J_+,J_-]}{\psi_0}=\matrixel{\psi_0}{\left(J_+\right)^2-\left(J_-\right)^2-2J_z}{\psi_0}&&\\
        &=2\left(c^*_{1}c_{-1}-c^*_{-1}c_1\right)-2\sum_{\mu}|c_{\mu}|^2\mu\quad\text{ (note that the last summation term is zero for our chosen probe)}\\
        &=2\left(c^*_{1}c_{-1}-c^*_{-1}c_1\right)=i\sin{(2\phi)}\bigg[1+\sin{(2\alpha)}\cos{(\psi)}\bigg].
    \end{flalign*}
    \begin{flalign*}
        &\text{\textbullet}\matrixel{\psi_0}{\left(J_+-J_-\right)\left(J_++J_-\right)}{\psi_0}=-\bigg(\matrixel{\psi_0}{\left(J_++J_-\right)\left(J_+-J_-\right)}{\psi_0}\bigg)^*=i\sin{(2\phi)}\bigg[1+\sin{(2\alpha)}\cos{(\psi)}\bigg].&&
    \end{flalign*}
    \begin{flalign*}
    &\text{\textbullet}\matrixel{\psi_0}{J_+J_z}{\psi_0}=-\sqrt{2}c^*_0c_{-1}\,;\,\matrixel{\psi_0}{J_-J_z}{\psi_0}=\sqrt{2}c^*_0c_1&&
    \end{flalign*}
    \begin{flalign*}
    &\text{\textbullet}\matrixel{\psi_0}{\left(J_+-J_-\right)J_z}{\psi_0}=-\sqrt{2}c^*_0\left(c_1+c_{-1}\right)\quad=-\cos{(\phi)}\bigg[\cos{(2\alpha)}+i\sin{(2\alpha)}\sin{(\psi)}\bigg].&&
    \end{flalign*}
    \begin{flalign*}
    &\text{\textbullet}\matrixel{\psi_0}{\left(J_++J_-\right)J_z}{\psi_0}=\sqrt{2}c^*_0\left(c_1-c_{-1}\right)=-i\sin{(\phi)}\bigg[\cos{(2\alpha)}+i\sin{(2\alpha)}\sin{(\psi)}\bigg].&&
    \end{flalign*}
    \begin{flalign*}
        &\text{\textbullet}\matrixel{\psi_0}{J_z\left(J_+-J_-\right)}{\psi_0}=-\bigg(\matrixel{\psi_0}{\left(J_+-J_-\right)J_z}={\psi_0}\bigg)^*=\sqrt{2}c_0\left(c^*_1+c^*_{-1}\right)=\cos{(\phi)}\bigg[\cos{(2\alpha)}-i\sin{(2\alpha)}\sin{(\psi)}\bigg].&&
    \end{flalign*}
    \begin{flalign*}
        &\text{\textbullet}\matrixel{\psi_0}{J_z\left(J_++J_-\right)}{\psi_0}=\bigg(\matrixel{\psi_0}{\left(J_++J_-\right)J_z}{\psi_0}\bigg)^*=\sqrt{2}c_0\left(c^*_1-c^*_{-1}\right)=i\sin{(\phi)}\bigg[\cos{(2\alpha)}-i\sin{(2\alpha)}\sin{(\psi)}\bigg].&&
    \end{flalign*}
    Now using the expressions derived above, we can find
    \begin{flalign*}
    &\text{Re}\bigg[\matrixel{\psi_0}{\left(\tilde{\bm{\eta}}_i\cdot\bm{J}\right)\left(\tilde{\bm{\eta}_j}\cdot\bm{J}\right)}{\psi_0}\bigg]=\tilde{\bm{\eta}}_i\cdot\tilde{\bm{\eta}}_j  -\frac{1}{2}\bigg[1+\sin{(2\alpha)}\cos{(\psi)}\bigg]\bigg[\left(\tilde{\eta}_i\right)_x\left(\tilde{\eta}_j\right)_x\sin^2{(\phi)}+\left(\tilde{\eta}_i\right)_y\left(\tilde{\eta}_j\right)_y\cos^2{(\phi)}\bigg]&&\\
    &-\frac{\left(\tilde{\eta}_i\right)_z\left(\tilde{\eta}_j\right)_z}{2}\bigg[1-\sin{(2\alpha)}\cos{(\psi)}\bigg]+\frac{\sin{(2\phi)}}{4}\bigg[1+\sin{(2\alpha)}\cos{(\psi)}\bigg]\bigg[\left(\tilde{\eta}_i\right)_x\left(\tilde{\eta}_j\right)_y+\left(\tilde{\eta}_i\right)_y\left(\tilde{\eta}_j\right)_x\bigg]\\&-\frac{\sin{(2\alpha)}\sin{(\psi)}\cos{(\phi)}}{2}\bigg[\left(\tilde{\eta}_i\right)_y\left(\tilde{\eta}_j\right)_z+\left(\tilde{\eta}_i\right)_z\left(\tilde{\eta}_j\right)_y\bigg]+\frac{\sin{(2\alpha)}\sin{(\psi)}\sin{(\phi)}}{2}\bigg[\left(\tilde{\eta}_i\right)_x\left(\tilde{\eta}_j\right)_z+\left(\tilde{\eta}_i\right)_z\left(\tilde{\eta}_j\right)_x\bigg].
    \end{flalign*}
    Thus, the QFIM elements would be given by
    \begin{flalign*}
        \left(\tilde{\mathcal{F}}_Q\right)_{ij}=4\bigg\lbrace{\text{Re}\bigg[\matrixel{\psi_0}{\left(\tilde{\bm{\eta}}_i\cdot\bm{J}\right)\left(\tilde{\bm{\eta}}_j\cdot\bm{J}\right)}{\psi_0}\bigg]-\cos^2{(2\alpha)}\left(\tilde{\bm{\eta}}_i\cdot\hat{n}_{rot}\right)(\tilde{\bm{\eta}}_j\cdot\hat{n}_{rot})\bigg\rbrace}.
    \end{flalign*}
    The expressions for $\tilde{\bm{\eta}}_{i,j}$ can be derived in exactly the same manner as that in the qubit-case using Wilcox's formula and they are given by
    \begin{align}
&\tilde{\bm{\eta}}_{\phi}=2\sin{\left(\frac{B}{2}\right)}\left(\cos{\left(\frac{B}{2}\right)\sin{(\phi)}},-\cos{\left(\frac{B}{2}\right)\cos{(\phi)}},\sin{\left(\frac{B}{2}\right)}\right)^T,&&\\
        &\tilde{\bm{\eta}}_{B}=-\left(\cos{(\phi)},\sin{(\phi)},0\right)^T.
\end{align}

\subsubsection{\texorpdfstring{Calculation of $\left(\mathcal{\tilde{F}}_Q\right)_{\phi,\phi}$}{Calculation of (Ftilde_Q)_(phi,phi)}}
In order to calculate $\left(\mathcal{\tilde{F}}_Q\right)_{\phi,\phi}$, which is the first diagonal element of the re-parameterized Quantum Fisher Information Matrix, we would need the following quantities: 
    \begin{align*}
\text{\textbullet}\left(\tilde{\eta}_{\phi}\right)_x^2&=\sin^2{(B)}\sin^2{(\phi)}&\text{\textbullet}2\left(\tilde{\eta}_{\phi}\right)_x\left(\tilde{\eta}_{\phi}\right)_y&=-\sin^2{(B)}\sin{(2\phi)}&&\\
\text{\textbullet}\left(\tilde{\eta}_{\phi}\right)^2_y&=\sin^2{(B)}\cos^2{(\phi)}&\text{\textbullet}2\left(\tilde{\eta}_{\phi}\right)_y\left(\tilde{\eta}_{\phi}\right)_z&=-4\sin{(B)}\cos{(\phi)}\sin^2{\left(\frac{B}{2}\right)}&&\\ \text{\textbullet}\left(\tilde{\eta}_{\phi}\right)^2_z&=4\sin^4{\left(\frac{B}{2}\right)}&\text{\textbullet}2\left(\tilde{\eta}_{\phi}\right)_x\left(\tilde{\eta}_{\phi}\right)_z&=4\sin{(B)}\sin{(\phi)}\sin^2{\left(\frac{B}{2}\right)}&&
    \end{align*}
    Additionally, we also have $|\tilde{\bm{\eta}}_{\phi}|^2=4\sin^2{\left(\frac{B}{2}\right)}$ and $\tilde{\bm{\eta}}_{\phi}\cdot\hat{n}_{rot}=0$ as we can see from the expressions of $\tilde{\bm{\eta}}_{\phi,B}$ above. Using these expressions, we can now evaluate the first diagonal element as
\begin{flalign*}  
&\implies \left(\tilde{F}_Q\right)_{\phi,\phi}\\&=4\bigg[4\sin^2{\left(\frac{B}{2}\right)}-\frac{\sin^2{\left(B\right)}}{2}\bigg(\sin^4{(\phi)}+\cos^4{(\phi)}\bigg)\bigg(1+\sin{(2\alpha)}\cos{(\psi)}\bigg)-2\sin^4{\left(\frac{B}{2}\right)}\bigg(1-\sin{(2\alpha)}\cos{(\psi)}\bigg),&&\\
&-\frac{\sin^2{(B)}\sin^2{(2\phi)}}{4}\bigg(1+\sin{(2\alpha)}\cos{(\psi)}\bigg)+2\sin{(B)}\sin^2{\left(\frac{B}{2}\right)}\sin{(\psi)}\sin{(2\alpha)}\bigg],&&\\
&=4\bigg[4\sin^2{\left(\frac{B}{2}\right)}-\frac{\sin^2{\left(B\right)}}{2}\bigg(1-\frac{\sin^2{(2\phi)}}{2}\bigg)\bigg(1+\sin{(2\alpha)}\cos{(\psi)}\bigg)-2\sin^4{\left(\frac{B}{2}\right)}\bigg(1-\sin{(2\alpha)}\cos{(\psi)}\bigg),&&\\
&-\frac{\sin^2{(B)}\sin^2{(2\phi)}}{4}\bigg(1+\sin{(2\alpha)}\cos{(\psi)}\bigg)+2\sin{(B)}\sin^2{\left(\frac{B}{2}\right)}\sin{(\psi)}\sin{(2\alpha)}\bigg]&&\\
&=4\bigg[4\sin^2{\left(\frac{B}{2}\right)}-2\sin^2{\left(\frac{B}{2}\right)}\cos^2{\left(\frac{B}{2}\right)}-\frac{\sin^2{(B)}\sin{(2\alpha)}\cos{(\psi)}}{2}-2\sin^4{\left(\frac{B}{2}\right)}+2\sin^4{\left(\frac{B}{2}\right)}\sin{(2\alpha)}\cos{(\psi)}&&\\&+2\sin{(B)}\sin^2{\left(\frac{B}{2}\right)}\sin{(\psi)}\sin{(2\alpha)}\bigg],&&\\
&=4\bigg[2\sin^2{\left(\frac{B}{2}\right)}-2\sin^2{\left(\frac{B}{2}\right)}\sin{(2\alpha)}\cos{(\psi)}\cos{(B)}+2\sin{(B)}\sin^2{\left(\frac{B}{2}\right)}\sin{(\psi)}\sin{(2\alpha)}\bigg],&&\\
&=8\sin^2{\left(\frac{B}{2}\right)}\bigg[1-\sin{(2\alpha)}\cos{(\psi+B)}\bigg].
\end{flalign*}

\subsubsection{\texorpdfstring{Calculation of $\left(\mathcal{\tilde{F}}_Q\right)_{B,B}$}{Calculation of (Ftilde_Q)_(B,B)}} 
Now we move on to calculate the second diagonal element of the re-parameterized QFIM, i.e., $\left(\tilde{\mathcal{F}}_Q\right)_{B,B}$. The quantities we would need are

\begin{flalign*}
\text{\textbullet}\left(\tilde{\eta}_B\right)^2_x&=\cos^2{(\phi)}&\text{\textbullet}\left(\tilde{\eta}_B\right)^2_y&=\sin^2{(\phi)}&\text{\textbullet}\left(\tilde{\eta}_B\right)_x\left(\tilde{\eta}_B\right)_y&=\sin{(\phi)}\cos{(\phi)}&&
\end{flalign*}
Additionally, $|\tilde{\bm{\eta}}_B|^2=1$ and $\tilde{\bm{\eta}}_B\cdot\hat{n}_{rot}=-1$
\begin{align*}
\implies \left(\tilde{\mathcal{F}}_Q\right)_{B,B}&=4\bigg[1-\bigg(1+\sin{(2\alpha)}\cos{(\psi)}\bigg)\sin^2{(\phi)}\cos^2{(\phi)}+\frac{1}{2}\sin{(2\phi)}\bigg(1+\sin{(2\alpha)}\cos{(\psi)}\bigg)\sin{(\phi)}\cos{(\phi)}-\cos^2{(2\alpha)}\bigg],\\
&=4\sin^2{(2\alpha)}.
\end{align*}

\subsubsection{\texorpdfstring{Calculation of $\left(\mathcal{\tilde{F}}_Q\right)_{\phi,B}$}{Calculation of (Ftilde_Q)_(phi,B)}}
Finally, we calculate the off-diagonal element of the QFIM. The required terms are listed below
\begin{align*} \text{\textbullet}\left(\tilde{\eta}_{\phi}\right)_x\left(\tilde{\eta}_{B}\right)_x&=-\sin{(B)}\sin{(\phi)}\cos{(\phi)}\hspace{2cm}\text{\textbullet}\left(\tilde{\eta}_{\phi}\right)_x\left(\tilde{\eta}_{B}\right)_y+\left(\tilde{\eta}_{\phi}\right)_y\left(\tilde{\eta}_{B}\right)_x=\sin{(B)}\cos{(2\phi)}&&\\
        \text{\textbullet}\left(\tilde{\eta}_{\phi}\right)_y\left(\tilde{\eta}_{B}\right)_y&=\sin{(B)}\sin{(\phi)}\cos{(\phi)}\hspace{2.35cm}\text{\textbullet}\left(\tilde{\eta}_{\phi}\right)_z\left(\tilde{\eta}_{B}\right)_y=-2\sin^2{\left(\frac{B}{2}\right)}\sin{(\phi)}\\
        \text{\textbullet}\left(\tilde{\eta}_{\phi}\right)_z\left(\tilde{\eta}_{B}\right)_x&=-2\sin^2{\left(\frac{B}{2}\right)}\cos{(\phi)}\hspace{2.35cm}\text{\textbullet}\,\,\tilde{\bm{\eta}}_{\phi}\cdot\tilde{\bm{\eta}}_B=0
    \end{align*}
    
    \begin{flalign*}    \implies\left(\tilde{F}_Q\right)_{\phi,B}&=4\bigg[0-\frac{1}{2}\bigg(1+\sin{(2\alpha)}\cos{(\psi)}\bigg)\bigg(-\sin{(B)}\sin^3{(\phi)}\cos{(\phi)}+\sin{(B)}\cos^3{(\phi)}\sin{(\phi)}\bigg)&&\\&+\frac{\sin{(2\phi)}\cos{(2\phi)}}{4}\sin{(B)}\bigg(1+\sin{(2\alpha)}\cos{(\psi)}\bigg)+\sin{(2\alpha)}\sin{(\psi)}\cos{(\phi)}\sin^2{\left(\frac{B}{2}\right)}\sin{(\phi)}&&\\&-\sin{(2\alpha)}\sin{(\psi)}\sin{(\phi)}\sin^2{\left(\frac{B}{2}\right)}\cos{(\phi)}\bigg],&&\\  &=0.
    \end{flalign*}

\end{document}